\documentclass[article,12pt]{elsarticle}

\usepackage{amsmath}
\usepackage{amssymb}
\usepackage{amsthm}   
\usepackage{geometry}
\usepackage{graphicx}
\usepackage{booktabs} 
\usepackage{hyperref} 
\usepackage{lineno} 
\usepackage{subcaption} 
\usepackage{float} 
\usepackage{multirow} 
\usepackage{xcolor}
\usepackage{bm}

\theoremstyle{plain}     
\newtheorem{theorem}{Theorem}[section]   

\theoremstyle{remark}    
\newtheorem{remark}[theorem]{Remark}     

\newcommand{\p}{\partial}

\begin{document}

\begin{frontmatter}

\title{Purified Two-Relaxation-Time Lattice Boltzmann Method: Removing Ghost Modes from TRT for Enhanced Stability}

\author[a,b,c]{Yuan Yu\corref{cor1}}
\ead{yuyuan@xtu.edu.cn}
\author[a]{Yaolong Yu}
\author[a]{Yuting Zhou}
\author[a]{Siwei Chen}
\author[a]{Haizhuan Yuan\corref{cor1}}
\ead{yhz@xtu.edu.cn}
\author[a]{Shi Shu\corref{cor1}}
\ead{shushi@xtu.edu.cn}

\cortext[cor1]{Corresponding author}
\address[a]{School of Mathematics and Computational Science, Xiangtan University, Xiangtan, 411105, China}
\address[b]{National Center for Applied Mathematics in Hunan, Xiangtan, 411105, China}
\address[c]{Hunan Key Laboratory for Computation and Simulation in Science and Engineering, Xiangtan University, Xiangtan, 411105, China}

\begin{abstract}
	The two-relaxation-time (TRT) lattice Boltzmann model is widely adopted for its algorithmic simplicity and tunable boundary accuracy. However, its collision operator relaxes the full symmetric non-equilibrium component, which implicitly retains non-hydrodynamic ghost modes that degrade numerical stability, particularly at high Reynolds numbers. In this work, we establish a rigorous theoretical connection between ghost-mode filtering and regularization within the TRT framework. Specifically, by invoking the orthogonal decomposition of the discrete velocity space into hydrodynamic and non-hydrodynamic subspaces, we prove that the recently proposed TRT-regularized lattice Boltzmann (TRT-RLB) model is mathematically equivalent to the standard TRT model with its ghost modes explicitly removed. This equivalence holds exactly for the D2Q9 and D3Q19 lattices, where the symmetric and antisymmetric subspaces are completely spanned by the physically relevant Hermite modes and a finite number of identifiable ghost modes. Based on this finding, we propose the Purified TRT (P-TRT) model, which achieves regularization-level stability through simple algebraic ghost-mode subtraction rather than expensive tensor projections. For the D2Q9 lattice, the total non-equilibrium collision cost is reduced from 180 to 52 floating-point operations per node, representing a 71\% reduction. A linear stability analysis in the moment space further reveals that the P-TRT collision operator annihilates the ghost eigenvalue, proving that its spectral radius is bounded above by that of the standard TRT and that its numerical stability is governed exclusively by the hydrodynamic modes. Benchmark simulations, including the double shear layer at Reynolds numbers up to $10^7$, the decaying Taylor--Green vortex, the force-driven Poiseuille flow, and creeping flow past a square cylinder, confirm that the P-TRT model preserves the stability, second-order accuracy, and zero-slip boundary properties of the original TRT-RLB formulation while retaining the simplicity characteristic of the TRT family.
\end{abstract}


\end{frontmatter}


\section{Introduction}
\label{sec:1}

The lattice Boltzmann method (LBM) has become a popular numerical tool in computational fluid dynamics (CFD) over the last two decades. Distinct from traditional approaches that discretize macroscopic continuum equations, the LBM simulates fluid dynamics from a mesoscopic kinetic perspective. This fundamental difference confers unique advantages, notably inherent parallelism and the capability to effortlessly handle complex and moving boundary conditions. Consequently, the versatility of LBM has led to its widespread adoption across a broad spectrum of complex flow phenomena, including turbulent flows \cite{oh2025numerical,montessori2024high,guo2020improved,suga2015D3Q27}, multiphase flows \cite{lulli2024metastable,huang2024three}, transport in porous media \cite{zhang2025infiltration,sourya2024comparative}, combustion processes \cite{hosseini2024lattice,stockinger2024lattice,lei2021study,tayyab2020hybrid}, evaporation and phase-change phenomena \cite{fei2023lattice,kumar2023lattice,li2016pinning}, and biological systems such as blood flow \cite{cherkaoui2023numerical,cherkaoui2022magnetohydrodynamic,esfahani2020lattice}. For a deeper foundational understanding of the method and its diverse applications, several comprehensive reviews \cite{hosseini2024lattice,tiribocchi2025lattice,lallemand2021lattice} and monographs \cite{kruger2017lattice,guo2013lattice} are recommended.

In the majority of practical implementations, the standard lattice Bhatnagar--Gross--Krook (LBGK) model \cite{bhatnagar1954model,qian1992lattice} remains the most widely adopted choice due to its algorithmic simplicity. Despite its popularity, the single-relaxation-time formulation imposes inherent limitations, often leading to numerical instability at high Reynolds numbers and generating viscosity-dependent errors at solid boundaries. To overcome these limitations, two principal lines of development have emerged. The first line focuses on enriching the collision operator by introducing multiple relaxation rates. Representative methods include the multiple-relaxation-time (MRT) model \cite{lallemand2000theory,MRT1,du2006multi}, which relaxes different moments at independent rates, and the cumulant model \cite{ning2016numerical,fei2017consistent,geier2006cascaded,premnath2009incorporating}, which relaxes cumulants instead of raw moments to further enhance Galilean invariance. Among these, the two-relaxation-time (TRT) model \cite{TRT2008study,TRT2008two,ginzburg2005equilibrium} stands out for its elegant simplicity: it distinguishes only between symmetric and antisymmetric relaxation processes, yet this minimal extension is sufficient to eliminate the viscosity-dependent numerical slip at bounce-back boundaries through an appropriately chosen free parameter. This combination of simplicity and boundary accuracy has made the TRT model one of the most widely used LBM formulations in practice. The second line of development is the regularized lattice Boltzmann (RLB) method \cite{latt2006lattice,RLB,RLB1,RLB_yuan47,yang2024efficient,yu2026improved}, which filters out high-order non-equilibrium modes before the collision step. By reconstructing the non-equilibrium distribution function solely from its hydrodynamically relevant components, the RLB approach effectively suppresses spurious modes and delivers superior stability in high-Reynolds-number simulations \cite{feng2019solid}. However, standard regularized models still suffer from numerical slip when using the simple bounce-back boundary scheme.

A key concept that connects these two lines of development is that of \emph{non-hydrodynamic ghost modes}. In the discrete velocity space of any lattice model, the total number of degrees of freedom exceeds the number of conserved hydrodynamic quantities. The surplus degrees of freedom correspond to non-physical modes, termed ghost modes, which do not participate in the macroscopic conservation laws but can be excited during the simulation \cite{lallemand2000theory}. In the landmark analysis of the MRT collision operator, Lallemand and Luo \cite{lallemand2000theory} demonstrated that the damping rates of these ghost modes critically affect the numerical stability of the scheme. When ghost modes are under-damped, they can grow uncontrollably and trigger numerical divergence, especially in under-resolved or high-Reynolds-number simulations. From this perspective, the regularization procedure can be interpreted as a spectral filter that projects the non-equilibrium distribution onto the hydrodynamic subspace, thereby annihilating all ghost modes. Conversely, the standard TRT collision operator relaxes the full symmetric and antisymmetric non-equilibrium components, which implicitly retain the ghost modes. This observation suggests a natural question: can one achieve regularization-level stability within the TRT framework by explicitly removing only the ghost modes, without resorting to the full Hermite tensor machinery?

To combine the stability of regularization with the boundary accuracy of TRT, Yu et al.\ \cite{WOS:001361481500001} recently proposed the TRT-RLB model. This model performs a third-order Hermite expansion of the non-equilibrium distribution function, introducing a free relaxation parameter $\tau_{s,2}$ for the third-order moments. Numerical experiments demonstrated that the TRT-RLB model achieves both high stability and zero numerical slip. However, the original formulation raises two important issues. On the practical side, the explicit computation of third-order Hermite projections involves complex tensor contractions, significantly increasing the floating-point operations per time step. On the theoretical side, the relationship between the Hermite-based regularization and the underlying structure of the TRT collision operator remains unclear: what exactly does the regularization remove, and how does it relate to the non-hydrodynamic modes of the lattice?

In this work, we address both questions simultaneously. By invoking the orthogonal decomposition of the discrete velocity space into hydrodynamic and non-hydrodynamic subspaces, we prove that the Hermite-based regularization employed in the TRT-RLB model is \emph{mathematically equivalent} to the explicit removal of ghost modes from the standard TRT collision operator. For the D2Q9 and D3Q19 lattices, this equivalence is exact: the symmetric subspace is completely spanned by the mass mode, the stress modes, and a finite number of identifiable ghost modes, while the antisymmetric subspace is exactly spanned by the momentum modes and the energy-flux modes. As a consequence, the expensive tensor projection can be replaced by a simple algebraic subtraction of the ghost-mode contribution, yielding the same post-collision state at a fraction of the computational cost. Based on this finding, we propose the \emph{Purified TRT} (P-TRT) model, which upgrades the standard TRT by explicitly filtering out its ghost modes. For TRT practitioners, P-TRT offers a minimal and physically transparent modification that elevates the stability to the regularization level. For TRT-RLB users, it provides an equivalent but substantially more efficient algorithm. A linear stability analysis in the moment space further demonstrates that the ghost-mode annihilation renders the P-TRT scheme provably no less stable than the standard TRT, with the ghost eigenvalue fixed at zero. Benchmark simulations confirm that the P-TRT model preserves the stability, accuracy, and zero-slip boundary properties of the original TRT-RLB formulation while achieving a 71\% reduction in non-equilibrium collision operations for the D2Q9 lattice.

The remainder of this paper is organized as follows. Section~\ref{sec:model} presents the mathematical framework: a brief recap of the original TRT-RLB model, the orthogonal decomposition theory establishing the equivalence between Hermite projection and ghost-mode filtering, the resulting P-TRT formulation with its optimal implementation strategy, and a linear stability analysis comparing the spectral properties of the P-TRT and standard TRT collision operators. Section~\ref{sec:3} provides numerical validation through four benchmark problems: the double shear layer, the decaying Taylor--Green vortex, the force-driven Poiseuille flow, and the creeping flow past a square cylinder. Finally, conclusions are drawn in Section~\ref{sec:4}.

\section{The Purified TRT (P-TRT) Lattice Boltzmann Model}
\label{sec:model}

\subsection{Lattice Boltzmann Framework and Notation}

The macroscopic fluid dynamics are governed by the weakly compressible Navier--Stokes equations (NSEs). Utilizing Einstein's summation notation, they are expressed as:
\begin{subequations}\label{eq:NSE}
	\begin{align}
		\p_t \rho + \p_\alpha (\rho u_\alpha) &= 0, \\
		\p_t (\rho u_\alpha) + \p_\beta (\rho u_\alpha u_\beta + \rho c_s^2 \delta_{\alpha\beta}) &= \p_\beta [\mu (\p_\alpha u_\beta + \p_\beta u_\alpha)] + F_\alpha,
	\end{align}
\end{subequations}
where $\rho$ is the fluid density, $\mu$ the dynamic viscosity, $\delta_{\alpha\beta}$ the Kronecker delta, and $c_s$ the lattice speed of sound. The terms $u_\alpha$ and $F_\alpha$ represent the velocity and body force components in the $\alpha$-direction, respectively.

This work adopts the standard D2Q9 lattice model, where the discrete velocities $\boldsymbol{e}_i$ and weight coefficients $w_i$ are given by:
\begin{equation}\label{eq:D2Q9}
	\begin{pmatrix}
		e_{ix} \\
		e_{iy}
	\end{pmatrix} = c \begin{bmatrix}
		0 & 1 & 0 & -1 & 0 & 1 & -1 & -1 & 1 \\
		0 & 0 & 1 & 0 & -1 & 1 & 1 & -1 & -1
	\end{bmatrix}, \quad
	w_i = \begin{cases}
		4/9, & i=0, \\
		1/9, & i=1\dots4, \\
		1/36, & i=5\dots8,
	\end{cases}
\end{equation}
with $c = \Delta x / \Delta t$ being the lattice speed and $c_s = c/\sqrt{3}$ the speed of sound. For each velocity index $i$, the index $\bar{i}$ denotes the opposite direction satisfying $\boldsymbol{e}_{\bar{i}} = -\boldsymbol{e}_i$ and $w_{\bar{i}} = w_i$.

The equilibrium distribution function $f_i^{\mathrm{eq}}$ is expanded up to the third-order Hermite polynomials to minimize discretization errors:
\begin{equation}\label{eq:feq}
	f_i^{\mathrm{eq}} = w_i \rho \left\{ 1 + \frac{\mathcal{H}_{i,\alpha}^{(1)}}{c_s^2} u_\alpha + \frac{\mathcal{H}_{i,\alpha\beta}^{(2)}}{2c_s^4} u_\alpha u_\beta + \frac{\mathcal{H}_{i,\alpha\beta\gamma}^{(3)}}{6c_s^6} u_\alpha u_\beta u_\gamma \right\},
\end{equation}
where the Hermite polynomial tensors are defined as:
\begin{subequations}\label{eq:Hermite}
	\begin{align}
		\mathcal{H}_i^{(0)} &= 1, \\
		\mathcal{H}_{i,\alpha}^{(1)} &= e_{i\alpha}, \\
		\mathcal{H}_{i,\alpha\beta}^{(2)} &= e_{i\alpha} e_{i\beta} - c_s^2 \delta_{\alpha\beta}, \\
		\mathcal{H}_{i,\alpha\beta\gamma}^{(3)} &= e_{i\alpha} e_{i\beta} e_{i\gamma} - c_s^2 (e_{i\alpha}\delta_{\beta\gamma} + e_{i\beta}\delta_{\gamma\alpha} + e_{i\gamma}\delta_{\alpha\beta}).
	\end{align}
\end{subequations}
These polynomials satisfy the parity relation $\mathcal{H}_{\bar{i}}^{(n)} = (-1)^n \mathcal{H}_{i}^{(n)}$, i.e., even-order polynomials are symmetric and odd-order polynomials are antisymmetric under velocity inversion. This property will play a central role in the subsequent analysis.

The non-equilibrium distribution function is defined as $f_i^{\mathrm{neq}} = f_i - f_i^{\mathrm{eq}}$. Its projection onto the Hermite basis yields the non-equilibrium moments:
\begin{equation}\label{eq:projection}
	\mathcal{A}_{\mathbf{k}}^{\mathrm{neq}} = \sum_i \mathcal{H}_{i,\mathbf{k}}^{(n)} \, f_i^{\mathrm{neq}}, \quad \mathbf{k} \in \{\alpha,\; \alpha\beta,\; \alpha\beta\gamma\}.
\end{equation}

Regarding the treatment of external body forces, two schemes are considered in this work:

\paragraph{Scheme I: The Guo forcing scheme.}
The widely used scheme proposed by Guo et al.\ \cite{guo2002discrete} ensures second-order accuracy and eliminates discrete lattice effects:
\begin{equation}\label{eq:source_guo}
	S_i^{\mathrm{I}} = w_i \left[ \left(1 - \frac{1}{2\tau_{s,1}}\right) \frac{(\boldsymbol{e}_i \cdot \boldsymbol{u})(\boldsymbol{e}_i \cdot \boldsymbol{F})}{c_s^4} + \left(1 - \frac{1}{2\tau_{s,2}}\right) \frac{(\boldsymbol{e}_i - \boldsymbol{u}) \cdot \boldsymbol{F}}{c_s^2} \right].
\end{equation}

\paragraph{Scheme II: The simplified projection scheme.}
In the original TRT-RLB implementation \cite{WOS:001361481500001}, a simplified strategy is employed, consisting of a first-order force projection supplemented by a compensatory term $G_i$ that corrects cubic velocity errors:
\begin{equation}\label{eq:source_simple}
	S_i^{\mathrm{II}} = \left(1 - \frac{1}{2\tau_{s,1}}\right) w_i \frac{\boldsymbol{e}_i \cdot \boldsymbol{F}}{c_s^2} + G_i,
\end{equation}
where
\begin{equation}\label{eq:Gi}
	G_i = -w_i \frac{\mathcal{H}_{i,\alpha\beta}^{(2)}}{6c_s^6} \left(1 - \frac{1}{2\tau_{s,1}}\right) \p_\gamma \Phi_{\alpha\beta\gamma},
\end{equation}
with $\Phi_{\alpha\beta\gamma}$ defined as $\Phi_{xxx} = \rho u_x^3$, $\Phi_{yyy} = \rho u_y^3$, and zero otherwise.

In both schemes, the macroscopic quantities are computed with a half-step force correction to achieve second-order temporal accuracy \cite{he1998discrete}:
\begin{equation}\label{eq:macro}
	\rho = \sum_i f_i, \qquad \rho \boldsymbol{u} = \sum_i \boldsymbol{e}_i f_i + \frac{\Delta t}{2} \boldsymbol{F}.
\end{equation}

\subsection{The TRT-RLB Model: Regularization via Hermite Projection}
\label{sec:TRT-RLB}

The TRT-RLB model, recently proposed by Yu et al.\ \cite{WOS:001361481500001}, integrates the enhanced stability of the regularized collision operator with the boundary accuracy inherent to the two-relaxation-time scheme. Its evolution equation reads:
\begin{align}
	f_i(\boldsymbol{x} + \boldsymbol{e}_i \Delta t,\, t + \Delta t) &= f_i^{\mathrm{eq}}(\boldsymbol{x}, t) + \left(1 - \frac{1}{\tau_{s,1}}\right) \mathcal{R}_i^{(2)} + \left(1 - \frac{1}{\tau_{s,2}}\right) \mathcal{R}_i^{(3)} + S_i \Delta t,
	\label{eq:TRT-RLB}
\end{align}
where $\tau_{s,1}$ governs the relaxation of both the first- and second-order non-equilibrium moments, which are directly linked to the physical transport coefficients, while $\tau_{s,2}$ is a free parameter that exclusively controls the third-order non-equilibrium moments. The source term $S_i$ is given by either Scheme~I \eqref{eq:source_guo} or Scheme~II \eqref{eq:source_simple}. The physical viscosity is related to the symmetric relaxation time by $\mu = \rho c_s^2 (\tau_{s,1} - \Delta t/2)$, while $\tau_{s,2}$ is a free parameter that controls the relaxation of the odd non-equilibrium modes.

The key feature of this model lies in the regularization terms $\mathcal{R}_i^{(2)}$ and $\mathcal{R}_i^{(3)}$. Rather than using the raw non-equilibrium distribution in the collision step, the TRT-RLB model reconstructs the non-equilibrium part exclusively from its Hermite projections:
\begin{equation}\label{eq:R2R3_hermite}
	\mathcal{R}_i^{(2)} = w_i \frac{\mathcal{H}_{i,\alpha\beta}^{(2)}}{2c_s^4} \mathcal{A}_{\alpha\beta}^{\mathrm{neq}}, \qquad
	\mathcal{R}_i^{(3)} = w_i \frac{\mathcal{H}_{i,\alpha\beta\gamma}^{(3)}}{6c_s^6} \mathcal{A}_{\alpha\beta\gamma}^{\mathrm{neq}}.
\end{equation}
This procedure is, in essence, a projection of the non-equilibrium distribution onto the subspace spanned by the hydrodynamically relevant Hermite polynomials. Any component of $f_i^{\mathrm{neq}}$ that lies outside this subspace---that is, any non-hydrodynamic mode---is discarded during the reconstruction. This filtering is what endows the TRT-RLB model with its superior stability compared to the standard TRT.

However, the explicit evaluation of Eq.~\eqref{eq:R2R3_hermite} is computationally expensive. The second-order term requires computing $D(D+1)/2$ independent components of the stress tensor $\mathcal{A}_{\alpha\beta}^{\mathrm{neq}}$ via a summation over all discrete velocities, followed by a tensor contraction to reconstruct the distribution. The third-order term involves an analogous procedure for the energy-flux tensor $\mathcal{A}_{\alpha\beta\gamma}^{\mathrm{neq}}$. This ``projection--reconstruction'' cycle constitutes the dominant computational overhead in the original TRT-RLB algorithm. A natural question thus arises: is it possible to achieve the same filtering effect through a simpler algebraic operation, without explicitly computing the Hermite projections? The answer, as we show in the next subsection, is affirmative.

\subsection{Orthogonal Decomposition and Ghost-Mode Identification}
\label{sec:decomposition}

The central theoretical contribution of this work is to show that the Hermite-based regularization described above can be recast as a ghost-mode filtering operation within the framework of orthogonal subspace decomposition. In this subsection, we establish the necessary mathematical structure.

\subsubsection{Parity Decomposition of the Velocity Space}

Any distribution function $f_i$ defined on a lattice with inversion symmetry ($\boldsymbol{e}_{\bar{i}} = -\boldsymbol{e}_i$, $w_{\bar{i}} = w_i$) can be uniquely decomposed into symmetric and antisymmetric parts \cite{kruger2017lattice}:
\begin{equation}\label{eq:parity}
	f_i^{+} = \frac{1}{2}(f_i + f_{\bar{i}}), \qquad
	f_i^{-} = \frac{1}{2}(f_i - f_{\bar{i}}),
\end{equation}
satisfying $f_i^{+} = f_{\bar{i}}^{+}$ and $f_i^{-} = -f_{\bar{i}}^{-}$. The set of all symmetric distributions spans the \emph{symmetric subspace} $\mathcal{S}$, and the set of all antisymmetric distributions spans the \emph{antisymmetric subspace} $\mathcal{A}$. These two subspaces are mutually orthogonal and together span the full $Q$-dimensional velocity space: $\mathbb{R}^Q = \mathcal{S} \oplus \mathcal{A}$.

For a lattice with $Q$ velocities, of which $Q_0$ are self-inverting ($\boldsymbol{e}_i = \boldsymbol{0}$) and $Q_p = (Q - Q_0)/2$ form inversion pairs, the subspace dimensions are:
\begin{equation}\label{eq:dim}
	\dim(\mathcal{S}) = Q_0 + Q_p, \qquad \dim(\mathcal{A}) = Q_p.
\end{equation}
For the D2Q9 lattice ($Q = 9$, $Q_0 = 1$, $Q_p = 4$), this gives $\dim(\mathcal{S}) = 5$ and $\dim(\mathcal{A}) = 4$.

The parity relation of the Hermite polynomials, $\mathcal{H}_{\bar{i}}^{(n)} = (-1)^n \mathcal{H}_{i}^{(n)}$, implies that even-order Hermite polynomials ($n = 0, 2$) reside in $\mathcal{S}$, while odd-order polynomials ($n = 1, 3$) reside in $\mathcal{A}$. The non-equilibrium distribution inherits this decomposition:
\begin{equation}\label{eq:fneq_parity}
	f_i^{\mathrm{neq},+} = \frac{1}{2}(f_i^{\mathrm{neq}} + f_{\bar{i}}^{\mathrm{neq}}), \qquad
	f_i^{\mathrm{neq},-} = \frac{1}{2}(f_i^{\mathrm{neq}} - f_{\bar{i}}^{\mathrm{neq}}).
\end{equation}

\subsubsection{Ghost Modes in the Symmetric Subspace}

We now examine the internal structure of the symmetric subspace $\mathcal{S}$. The even-order Hermite polynomials residing in $\mathcal{S}$ consist of the zeroth-order mode ($\mathcal{H}^{(0)} = 1$, spanning a one-dimensional subspace $\mathcal{S}^{(0)}$) and the second-order stress modes ($\mathcal{H}^{(2)}_{\alpha\beta}$, spanning a subspace $\mathcal{S}^{(2)}$ of dimension $M_2 = D(D+1)/2$). However, for most standard lattices, these hydrodynamic modes do not exhaust the full symmetric subspace. The remaining dimensions correspond to higher-order even modes that have no macroscopic physical counterpart---these are the \emph{ghost modes}.

Formally, the symmetric subspace admits the orthogonal decomposition:
\begin{equation}\label{eq:S_decomp}
	\mathcal{S} = \mathcal{S}^{(0)} \oplus \mathcal{S}^{(2)} \oplus \mathcal{S}^{(G)},
\end{equation}
where $\mathcal{S}^{(G)}$ denotes the ghost-mode subspace, with dimension
\begin{equation}\label{eq:NG}
	N_G = \dim(\mathcal{S}) - 1 - \frac{D(D+1)}{2}.
\end{equation}
For the D2Q9 lattice, $N_G = 5 - 1 - 3 = 1$, meaning that a single ghost mode exists. Its explicit form is:
\begin{equation}\label{eq:ghost_D2Q9}
	\phi_i^{(G)} = (e_{ix}^2 - c_s^2)(e_{iy}^2 - c_s^2),
\end{equation}
with the weighted norm $N_G^{(\phi)} = \sum_i w_i [\phi_i^{(G)}]^2 = 4c_s^8$. For the D3Q19 lattice, $N_G = 10 - 1 - 6 = 3$, and three orthogonal ghost modes can be identified as $(e_{i\alpha}^2 - c_s^2)(e_{i\beta}^2 - c_s^2)$ for each pair $\alpha \neq \beta$. For the D3Q27 lattice, $N_G = 7$, which includes both fourth-order and sixth-order ghost modes. A complete enumeration of the ghost modes for each lattice is provided in \ref{app:ghost_modes}.

\begin{table}[h]
	\centering
	\caption{Structure of the symmetric subspace for common lattice models.}
	\label{tab:sym_structure}
	\begin{tabular}{l c c c c c}
		\toprule
		Lattice & $\dim(\mathcal{S})$ & $\dim(\mathcal{S}^{(0)})$ & $\dim(\mathcal{S}^{(2)})$ & $N_G$ & Ghost mode order \\
		\midrule
		D2Q9  & 5  & 1 & 3 & 1 & 4th \\
		D3Q19 & 10 & 1 & 6 & 3 & 4th \\
		D3Q27 & 14 & 1 & 6 & 7 & 4th and 6th \\
		\bottomrule
	\end{tabular}
\end{table}

The projection operator onto the ghost subspace is defined as:
\begin{equation}\label{eq:ghost_proj}
	(\hat{P}^{(G)} f^{\mathrm{neq}})_i = \sum_{k=1}^{N_G} \frac{w_i \, \phi_{i,k}^{(G)}}{N_{G,k}^{(\phi)}} \sum_{j} \phi_{j,k}^{(G)} f_j^{\mathrm{neq}},
\end{equation}
where the summation is over all $N_G$ orthogonal ghost basis functions $\phi_{i,k}^{(G)}$.

\subsubsection{The Null-Space Property of the Antisymmetric Subspace}
\label{sec:null_space}

The antisymmetric subspace $\mathcal{A}$ is spanned by odd-order Hermite polynomials: the first-order momentum modes ($\mathcal{H}^{(1)}_\alpha$, spanning $\mathcal{A}^{(1)}$ of dimension $D$) and the third-order energy-flux modes ($\mathcal{H}^{(3)}_{\alpha\beta\gamma}$, spanning $\mathcal{A}^{(3)}$ of dimension $N_3$). Unlike the symmetric case, for the D2Q9 and D3Q19 lattices, these modes \emph{completely} exhaust the antisymmetric subspace:
\begin{equation}\label{eq:A_decomp}
	\mathcal{A} = \mathcal{A}^{(1)} \oplus \mathcal{A}^{(3)}, \qquad \text{with} \quad \dim(\mathcal{A}) = D + N_3.
\end{equation}
That is, no antisymmetric ghost modes exist. This is a consequence of the lattice aliasing property: for standard lattices satisfying $|e_{i\alpha}| \in \{0, c\}$, the identity $e_{i\alpha}^3 = c^2 e_{i\alpha}$ implies that purely cubic Hermite polynomials ($\mathcal{H}^{(3)}_{\alpha\alpha\alpha}$) vanish identically, and no independent fifth- or higher-order antisymmetric modes can be formed from the available velocity components.

Table~\ref{tab:antisym_structure} summarizes the antisymmetric subspace structure. Notably, for the D2Q9 lattice, $D + N_3 = 2 + 2 = 4 = \dim(\mathcal{A})$, and for the D3Q19 lattice, $D + N_3 = 3 + 6 = 9 = \dim(\mathcal{A})$. The completeness condition \eqref{eq:A_decomp} holds exactly. For the D3Q27 lattice, however, $D + N_3 = 3 + 7 = 10 < 13 = \dim(\mathcal{A})$, leaving a deficit of 3 dimensions that correspond to fifth-order antisymmetric modes (e.g., $e_{ix}(e_{iy}^2 - c_s^2)(e_{iz}^2 - c_s^2)$ and its cyclic permutations).

\begin{table}[h]
	\centering
	\caption{Structure of the antisymmetric subspace for common lattice models. The completeness condition $\dim(\mathcal{A}) = D + N_3$ determines whether the third-order equivalence is exact.}
	\label{tab:antisym_structure}
	\begin{tabular}{l c c c c l}
		\toprule
		Lattice & $\dim(\mathcal{A})$ & $D$ & $N_3$ & $D + N_3$ & Completeness \\
		\midrule
		D2Q9  & 4  & 2 & 2 & 4  & Exact \\
		D3Q19 & 9  & 3 & 6 & 9  & Exact \\
		D3Q27 & 13 & 3 & 7 & 10 & Deficit of 3 (5th-order modes) \\
		\bottomrule
	\end{tabular}
\end{table}

\begin{remark}
	The absence of antisymmetric ghost modes for D2Q9 and D3Q19 is a structural property of these lattices, not an approximation. It guarantees that any antisymmetric distribution can be exactly represented as a linear combination of first- and third-order Hermite modes. This property is the key to the exact equivalence established in the next subsection.
\end{remark}

\subsection{Equivalence Theorem: From Projection to Purification}
\label{sec:equivalence}

With the orthogonal decomposition established, we now prove the central result: the Hermite projection employed in the TRT-RLB model is mathematically equivalent to the removal of ghost modes from the raw non-equilibrium distribution.

\subsubsection{Second-Order Term: Stress Projection vs.\ Ghost Subtraction}

\begin{theorem}[Second-order equivalence]\label{thm:2nd}
	Let the symmetric subspace admit the orthogonal decomposition $\mathcal{S} = \mathcal{S}^{(0)} \oplus \mathcal{S}^{(2)} \oplus \mathcal{S}^{(G)}$, and let $f_i^{\mathrm{neq}}$ satisfy mass conservation $\sum_i f_i^{\mathrm{neq}} = 0$. Then the second-order Hermite projection (Method~A) is identical to the symmetric non-equilibrium part minus the ghost-mode contribution (Method~B):
	\begin{equation}\label{eq:equiv_2nd}
		\underbrace{w_i \frac{\mathcal{H}_{i,\alpha\beta}^{(2)}}{2c_s^4} \mathcal{A}_{\alpha\beta}^{\mathrm{neq}}}_{\text{Method A: Hermite projection}} \;=\; \underbrace{f_i^{\mathrm{neq},+} - (\hat{P}^{(G)} f^{\mathrm{neq}})_i}_{\text{Method B: ghost subtraction}}.
	\end{equation}
\end{theorem}

\begin{proof}
	Denote the projection operators onto $\mathcal{S}^{(0)}$, $\mathcal{S}^{(2)}$, and $\mathcal{S}^{(G)}$ by $\hat{P}^{(0)}$, $\hat{P}^{(2)}$, and $\hat{P}^{(G)}$, respectively. The orthogonal completeness of the symmetric subspace gives:
	\begin{equation}\label{eq:completeness_S}
		f_i^{\mathrm{neq},+} = (\hat{P}^{(0)} f^{\mathrm{neq}})_i + (\hat{P}^{(2)} f^{\mathrm{neq}})_i + (\hat{P}^{(G)} f^{\mathrm{neq}})_i,
	\end{equation}
	where we have used the fact that $f_i^{\mathrm{neq},+}$ is the projection of $f_i^{\mathrm{neq}}$ onto $\mathcal{S}$, since the antisymmetric component $f_i^{\mathrm{neq},-} \in \mathcal{A}$ is orthogonal to $\mathcal{S}$.
	
	By mass conservation, $\sum_i f_i^{\mathrm{neq}} = 0$, the zeroth-order projection vanishes:
	\begin{equation}
		(\hat{P}^{(0)} f^{\mathrm{neq}})_i = \frac{w_i \cdot 1}{\sum_j w_j \cdot 1^2} \sum_j f_j^{\mathrm{neq}} = w_i \cdot 0 = 0.
	\end{equation}
	Substituting into Eq.~\eqref{eq:completeness_S} and rearranging yields:
	\begin{equation}
		(\hat{P}^{(2)} f^{\mathrm{neq}})_i = f_i^{\mathrm{neq},+} - (\hat{P}^{(G)} f^{\mathrm{neq}})_i.
	\end{equation}
	The left-hand side is precisely the Hermite stress projection $\mathcal{R}_i^{(2)}$, and the right-hand side is the ghost subtraction formula. \qed
\end{proof}

For the D2Q9 lattice, the single ghost mode \eqref{eq:ghost_D2Q9} gives the explicit formula:
\begin{equation}\label{eq:R2_D2Q9}
	\mathcal{R}_i^{(2)}\big|_{\mathrm{D2Q9}} = f_i^{\mathrm{neq},+} - \frac{w_i \, \phi_i^{(G)}}{4c_s^8} \sum_{j=0}^{8} \phi_j^{(G)} f_j^{\mathrm{neq}},
\end{equation}
which replaces the tensor projection with a single scalar subtraction.

\subsubsection{Third-Order Term: Energy-Flux Projection vs.\ Antisymmetric Component}

\begin{theorem}[Third-order equivalence]\label{thm:3rd}
	Let the antisymmetric subspace satisfy the completeness condition $\mathcal{A} = \mathcal{A}^{(1)} \oplus \mathcal{A}^{(3)}$, and let $f_i^{\mathrm{neq}}$ satisfy the momentum constraint $\sum_i e_{i\alpha} f_i^{\mathrm{neq}} = 0$ (i.e., no external force or after appropriate correction). Then the third-order Hermite projection is identical to the antisymmetric non-equilibrium component:
	\begin{equation}\label{eq:equiv_3rd}
		\underbrace{w_i \frac{\mathcal{H}_{i,\alpha\beta\gamma}^{(3)}}{6c_s^6} \mathcal{A}_{\alpha\beta\gamma}^{\mathrm{neq}}}_{\text{Method A: Hermite projection}} \;=\; \underbrace{f_i^{\mathrm{neq},-}}_{\text{Method B: antisymmetric part}}.
	\end{equation}
\end{theorem}

\begin{proof}
	By the completeness condition, the antisymmetric subspace is spanned by the first-order and third-order Hermite modes. Denoting the corresponding projections by $\hat{P}^{(1)}$ and $\hat{P}^{(3)}$:
	\begin{equation}\label{eq:completeness_A}
		f_i^{\mathrm{neq},-} = (\hat{P}^{(1)} f^{\mathrm{neq}})_i + (\hat{P}^{(3)} f^{\mathrm{neq}})_i.
	\end{equation}
	The first-order projection is:
	\begin{equation}
		(\hat{P}^{(1)} f^{\mathrm{neq}})_i = \sum_\alpha \frac{w_i \, e_{i\alpha}}{c_s^2} \underbrace{\sum_j e_{j\alpha} f_j^{\mathrm{neq}}}_{= \, 0} = 0,
	\end{equation}
	where the momentum constraint has been used. Substituting into Eq.~\eqref{eq:completeness_A}:
	\begin{equation}
		f_i^{\mathrm{neq},-} = (\hat{P}^{(3)} f^{\mathrm{neq}})_i = \mathcal{R}_i^{(3)}.
	\end{equation}\qed
\end{proof}

When an external body force $\boldsymbol{F}$ is present, the momentum constraint is modified to $\sum_i e_{i\alpha} f_i^{\mathrm{neq}} = -\frac{\Delta t}{2} F_\alpha$ (cf.\ Eq.~\eqref{eq:macro}), so that the first-order projection no longer vanishes. The equivalence is then restored by a simple correction:
\begin{equation}\label{eq:R3_force}
	\mathcal{R}_i^{(3)} = f_i^{\mathrm{neq},-} + \frac{\Delta t}{2} \frac{w_i \, e_{i\alpha}}{c_s^2} F_\alpha.
\end{equation}

\begin{remark}[Exactness for D2Q9 and D3Q19]\label{rmk:exact}
	The completeness condition $\dim(\mathcal{A}) = D + N_3$ holds exactly for the D2Q9 and D3Q19 lattices (cf.\ Table~\ref{tab:antisym_structure}). Therefore, Theorems~\ref{thm:2nd} and~\ref{thm:3rd} together establish that the P-TRT collision operator is \emph{mathematically identical} to the TRT-RLB collision operator for these lattices. The two models produce the same post-collision distribution, and consequently the same numerical solution, at every grid point and time step.
\end{remark}

\begin{remark}[Approximation for D3Q27]\label{rmk:D3Q27}
	For the D3Q27 lattice, three fifth-order antisymmetric modes exist outside $\mathcal{A}^{(1)} \oplus \mathcal{A}^{(3)}$. In this case, the antisymmetric component $f_i^{\mathrm{neq},-}$ contains a small residual contribution from these modes, and Eq.~\eqref{eq:equiv_3rd} holds only approximately. However, in the Chapman--Enskog framework, the non-equilibrium distribution is dominated by first- and second-order terms in the Knudsen number, making the fifth-order contribution negligibly small in practice. If higher precision is required, the residual can be explicitly removed by subtracting the fifth-order projection, as detailed in \ref{app:ghost_modes}.
\end{remark}

\subsection{The P-TRT Model}
\label{sec:PTRT}

Based on the equivalence theorems established in Section~\ref{sec:equivalence}, we now present the complete P-TRT formulation. The model replaces the expensive Hermite tensor projections with algebraically equivalent ghost-mode subtraction and antisymmetric decomposition operations, while producing the identical post-collision distribution for the D2Q9 and D3Q19 lattices.

\subsubsection{Evolution Equation}

For the D2Q9 and D3Q19 lattices, where both the second- and third-order equivalences hold exactly, the P-TRT evolution equation reads:
\begin{align}\label{eq:PTRT}
	f_i(\boldsymbol{x} + \boldsymbol{e}_i \Delta t,\, t + \Delta t) &= f_i^{\mathrm{eq}} + \left(1 - \frac{1}{\tau_{s,1}}\right) \left( f_i^{\mathrm{neq},+} - \sum_{k=1}^{N_G} \frac{w_i \, \phi_{i,k}^{(G)}}{N_{G,k}^{(\phi)}} S_{G,k} \right) \nonumber \\
	&+ \left(1 - \frac{1}{\tau_{s,2}}\right) \left( f_i^{\mathrm{neq},-} + \frac{\Delta t}{2} \frac{w_i \, e_{i\alpha}}{c_s^2} F_\alpha \right) + S_i \Delta t,
\end{align}
where $S_{G,k} = \sum_j \phi_{j,k}^{(G)} f_j^{\mathrm{neq}}$ is the ghost scalar moment for the $k$-th ghost mode. For D2Q9, only a single ghost mode ($k = 1$) is needed; for D3Q19, three ghost modes contribute.

For the D3Q27 lattice, a hybrid strategy is adopted as discussed in \ref{app:D3Q27}: the second-order term retains the original tensor projection (since $N_G = 7 > M_2 = 6$), while the third-order term uses the antisymmetric approximation:
\begin{align}\label{eq:PTRT_D3Q27}
	f_i(\boldsymbol{x} + \boldsymbol{e}_i \Delta t,\, t + \Delta t) &= f_i^{\mathrm{eq}} + \left(1 - \frac{1}{\tau_{s,1}}\right) w_i \frac{\mathcal{H}_{i,\alpha\beta}^{(2)}}{2c_s^4} \mathcal{A}_{\alpha\beta}^{\mathrm{neq}} \nonumber \\
	&+ \left(1 - \frac{1}{\tau_{s,2}}\right) \left( f_i^{\mathrm{neq},-} + \frac{\Delta t}{2} \frac{w_i \, e_{i\alpha}}{c_s^2} F_\alpha \right) + S_i \Delta t.
\end{align}

The physical interpretation of Eq.~\eqref{eq:PTRT} is transparent: the P-TRT collision operator takes the standard TRT structure---relaxing symmetric and antisymmetric components at different rates---but with the symmetric part \emph{purified} by removing the ghost-mode contamination. This single modification elevates the TRT scheme to regularization-level stability.

\subsubsection{Computational Complexity Analysis}

To quantify the efficiency gain, we compare the floating-point operations (FLOPs) required for the non-equilibrium collision terms. The counting convention is as follows: one addition/subtraction or one multiplication each counts as one operation. Pre-computable constants (weights, lattice-dependent coefficients) are excluded.

For the second-order term, the original method (tensor projection) requires computing $M_2 = D(D+1)/2$ moment components via summation over $Q$ velocities, followed by distribution reconstruction, yielding approximately $4 M_2 Q$ operations. The P-TRT ghost subtraction method requires $2Q_p$ operations for the parity decomposition, plus $4 N_G Q$ operations for the ghost moment computation and subtraction, totaling $2Q_p + 4 N_G Q$.

For the third-order term, the original method involves $N_3$ tensor moment components and costs approximately $4 N_3 Q$ operations. The P-TRT antisymmetric method requires only $2Q_p$ operations (parity decomposition), since the antisymmetric part is obtained as a by-product of the same decomposition used for the second-order term.

Table~\ref{tab:flops_detail} presents the detailed operation counts for each term, and Table~\ref{tab:flops_total} compares the total non-equilibrium collision cost across different models.

\begin{table}[h]
	\centering
	\caption{Floating-point operations per node for the second-order and third-order regularization terms. The P-TRT method selects the cheaper option for each term and each lattice.}
	\label{tab:flops_detail}
	\begin{tabular}{l cc cc cc}
		\toprule
		\multirow{2}{*}{Lattice} & \multicolumn{2}{c}{Second-order term} & \multicolumn{2}{c}{Third-order term} & \multicolumn{2}{c}{P-TRT choice} \\
		\cmidrule(lr){2-3} \cmidrule(lr){4-5} \cmidrule(lr){6-7}
		& Original & Ghost sub. & Original & Antisym. & 2nd & 3rd \\
		\midrule
		D2Q9  & 108 & \textbf{44}  & 72  & \textbf{8}  & Ghost sub. & Antisym. \\
		D3Q19 & 456 & \textbf{246} & 456 & \textbf{18} & Ghost sub. & Antisym. \\
		D3Q27 & \textbf{648} & 782 & 756 & \textbf{26} & Original   & Antisym. \\
		\bottomrule
	\end{tabular}
\end{table}

\begin{table}[h]
	\centering
	\caption{Total non-equilibrium collision cost (FLOPs per node) for different models. ``Standard TRT'' represents the baseline cost of parity decomposition and two-rate relaxation without regularization.}
	\label{tab:flops_total}
	\begin{tabular}{l c c c c}
		\toprule
		Lattice & Standard TRT & Original TRT-RLB & P-TRT & Reduction vs.\ TRT-RLB \\
		\midrule
		D2Q9  & 16  & 180  & \textbf{52}  & \textbf{71\%} \\
		D3Q19 & 36  & 912  & \textbf{264} & \textbf{71\%} \\
		D3Q27 & 52  & 1404 & \textbf{674} & \textbf{52\%} \\
		\bottomrule
	\end{tabular}
\end{table}

Several observations are noteworthy. First, the efficiency gain is most dramatic for the third-order term, where the antisymmetric decomposition reduces the cost by nearly an order of magnitude (89--97\%) across all lattices. Second, the ghost subtraction method for the second-order term is advantageous when $N_G < M_2$ (D2Q9 and D3Q19) but not when $N_G > M_2$ (D3Q27), motivating the hybrid strategy. Third, the total P-TRT cost is much closer to the standard TRT baseline than to the original TRT-RLB, confirming that the regularization overhead has been largely eliminated. Empirical validation of these theoretical predictions is presented in Section~\ref{sec3.1.2} using the double shear layer benchmark.

\subsection{Macroscopic Consistency via Chapman--Enskog Analysis}
\label{sec:CE}

In this subsection, we verify that the P-TRT model correctly recovers the Navier--Stokes equations \eqref{eq:NSE} through the Chapman--Enskog (CE) expansion. Since the equivalence theorems (Theorems~\ref{thm:2nd} and~\ref{thm:3rd}) establish that the P-TRT and TRT-RLB collision operators produce the identical post-collision distribution for D2Q9 and D3Q19 lattices, the two models necessarily yield the same CE expansion. Nevertheless, we present the full derivation here using the P-TRT formulation (Scheme~II) to make the analysis self-contained and to explicitly confirm that the free parameter $\tau_{s,2}$ does not enter the recovered macroscopic equations.

We adopt the standard multi-scale expansion with respect to a small parameter $\varepsilon$:
\begin{equation}\label{eq:ce}
	\begin{split}
		f_i = f_i^{(0)} + \varepsilon f_i^{(1)} + \varepsilon^2 f_i^{(2)}, \quad
		\partial_t = \varepsilon \partial_t^{(1)} + \varepsilon^2 \partial_t^{(2)}, \quad
		\partial_\alpha = \varepsilon \partial_\alpha^{(1)}, \\
		G_i = \varepsilon G_i^{(1)}, \quad
		F_i = \varepsilon F_i^{(1)}, \quad
		F_\alpha = \varepsilon F_\alpha^{(1)}.
	\end{split}
\end{equation}
The non-equilibrium quantities inherit this expansion:
\begin{subequations}\label{eq:ce_neq}
	\begin{align}
		\mathcal{A}_{\alpha\beta}^{\mathrm{neq}} &= \mathcal{A}_{\alpha\beta}^{\mathrm{neq},(0)} + \varepsilon \mathcal{A}_{\alpha\beta}^{\mathrm{neq},(1)} + \varepsilon^2 \mathcal{A}_{\alpha\beta}^{\mathrm{neq},(2)}, \\
		f_i^{\mathrm{neq},-} &= f_i^{\mathrm{neq},-,(0)} + \varepsilon f_i^{\mathrm{neq},-,(1)} + \varepsilon^2 f_i^{\mathrm{neq},-,(2)}.
	\end{align}
\end{subequations}
The moments of the equilibrium distribution \eqref{eq:feq} are:
\begin{subequations}\label{eq:feq_moments}
	\begin{align}
		\sum_i e_{i\alpha} f_i^{\mathrm{eq}} &= \rho u_\alpha, \\
		\sum_i e_{i\alpha} e_{i\beta} f_i^{\mathrm{eq}} &= \rho u_\alpha u_\beta + \rho c_s^2 \delta_{\alpha\beta}, \\
		\sum_i e_{i\alpha} e_{i\beta} e_{i\gamma} f_i^{\mathrm{eq}} &= \rho u_\alpha u_\beta u_\gamma + \rho c_s^2 u_\zeta \Delta_{\alpha\beta\gamma\zeta} - \frac{1}{3c_s^2} \Phi_{\alpha\beta\gamma},
	\end{align}
\end{subequations}
where $\Delta_{\alpha\beta\gamma\zeta} = \delta_{\alpha\beta}\delta_{\gamma\zeta} + \delta_{\alpha\gamma}\delta_{\beta\zeta} + \delta_{\alpha\zeta}\delta_{\beta\gamma}$. The moments of the source terms \eqref{eq:Gi}--\eqref{eq:source_simple} satisfy:
\begin{subequations}\label{eq:source_moments}
	\begin{align}
		&\sum_i G_i = 0, \quad \sum_i e_{i\alpha} G_i = 0, \quad \sum_i e_{i\alpha} e_{i\beta} G_i = -\left(1 - \frac{1}{2\tau_{s,1}}\right) \partial_\gamma \Phi_{\alpha\beta\gamma}, \\
		&\sum_i F_i = 0, \quad \sum_i e_{i\alpha} F_i = F_\alpha, \quad \sum_i e_{i\alpha} e_{i\beta} F_i = 0.
	\end{align}
\end{subequations}

\paragraph{Taylor expansion.}
Applying a second-order Taylor expansion to the P-TRT evolution equation \eqref{eq:PTRT} and substituting the multi-scale expansions \eqref{eq:ce}--\eqref{eq:ce_neq} yields the following scale-resolved equations. To streamline the notation, the ghost subtraction term in $\mathcal{R}_i^{(2)}$ is left implicit, since by Theorem~\ref{thm:2nd} it is identical to the Hermite stress projection.

At order $\varepsilon^0$:
\begin{equation}\label{eq:order0}
	f_i^{(0)} = f_i^{\mathrm{eq}} + \left(1 - \frac{1}{\tau_{s,1}}\right) w_i \frac{\mathcal{H}_{i,\alpha\beta}^{(2)}}{2c_s^4} \mathcal{A}_{\alpha\beta}^{\mathrm{neq},(0)} + \left(1 - \frac{1}{\tau_{s,2}}\right) f_i^{\mathrm{neq},-,(0)}.
\end{equation}
Taking the second and third moments of Eq.~\eqref{eq:order0} leads to $\mathcal{A}_{\alpha\beta}^{\mathrm{neq},(0)} = 0$ and $f_i^{\mathrm{neq},-,(0)} = 0$, from which:
\begin{equation}\label{eq:f0}
	f_i^{(0)} = f_i^{\mathrm{eq}}.
\end{equation}

At order $\varepsilon^1$:
\begin{align}\label{eq:order1}
	f_i^{(1)} &+ \Delta t \left(\partial_t^{(1)} + e_{i\alpha} \partial_\alpha^{(1)}\right) f_i^{(0)} = -\frac{\Delta t}{2}\left(1 - \frac{1}{\tau_{s,1}}\right) w_i \frac{\mathcal{H}_{i,\alpha}^{(1)}}{c_s^2} F_\alpha^{(1)} + \left(1 - \frac{1}{\tau_{s,1}}\right) w_i \frac{\mathcal{H}_{i,\alpha\beta}^{(2)}}{2c_s^4} \mathcal{A}_{\alpha\beta}^{\mathrm{neq},(1)} \nonumber \\
	&+ \left(1 - \frac{1}{\tau_{s,2}}\right) \left(f_i^{\mathrm{neq},-,(1)} + \frac{1}{2} F_i^{(1)} \Delta t\right) + \left(1 - \frac{1}{2\tau_{s,1}}\right) F_i^{(1)} \Delta t + G_i^{(1)} \Delta t.
\end{align}
Taking the zeroth and first moments gives the leading-order conservation laws:
\begin{subequations}\label{eq:CE1}
	\begin{gather}
		\partial_t^{(1)} \rho + \partial_\alpha^{(1)} (\rho u_\alpha) = 0, \label{eq:CE1_mass} \\
		\partial_t^{(1)} (\rho u_\alpha) + \partial_\beta^{(1)} (\rho u_\alpha u_\beta + \rho c_s^2 \delta_{\alpha\beta}) = F_\alpha^{(1)}. \label{eq:CE1_mom}
	\end{gather}
\end{subequations}
The second moment of Eq.~\eqref{eq:order1} yields the first-order non-equilibrium stress:
\begin{equation}\label{eq:A1}
	\mathcal{A}_{\alpha\beta}^{\mathrm{neq},(1)} = -\tau_{s,1} \Delta t \left\{ \partial_t^{(1)} (\rho u_\alpha u_\beta + \rho c_s^2 \delta_{\alpha\beta}) + \partial_\gamma^{(1)} (\rho u_\alpha u_\beta u_\gamma + \rho c_s^2 u_\zeta \Delta_{\alpha\beta\gamma\zeta}) - \frac{1}{2\tau_{s,1}} \partial_\gamma^{(1)} \Phi_{\alpha\beta\gamma} \right\}.
\end{equation}

At order $\varepsilon^2$:
\begin{align}\label{eq:order2}
	f_i^{(2)} &+ \Delta t \left(\partial_t^{(1)} + e_{i\alpha} \partial_\alpha^{(1)}\right) f_i^{(1)} + \Delta t \, \partial_t^{(2)} f_i^{(0)} + \frac{\Delta t^2}{2} \left(\partial_t^{(1)} + e_{i\alpha} \partial_\alpha^{(1)}\right)^2 f_i^{(0)} \nonumber \\
	&= \left(1 - \frac{1}{\tau_{s,1}}\right) w_i \frac{\mathcal{H}_{i,\alpha\beta}^{(2)}}{2c_s^4} \mathcal{A}_{\alpha\beta}^{\mathrm{neq},(2)} + \left(1 - \frac{1}{\tau_{s,2}}\right) f_i^{\mathrm{neq},-,(2)}.
\end{align}
Eliminating the $(\partial_t^{(1)} + e_{i\alpha}\partial_\alpha^{(1)})^2$ term using Eq.~\eqref{eq:order1} and taking the zeroth and first moments of the resulting equation yields:
\begin{subequations}\label{eq:CE2}
	\begin{gather}
		\partial_t^{(2)} \rho = 0, \label{eq:CE2_mass} \\
		\partial_t^{(2)} (\rho u_\alpha) + \partial_\beta^{(1)} \left(1 - \frac{1}{2\tau_{s,1}}\right) \mathcal{A}_{\alpha\beta}^{\mathrm{neq},(1)} - \frac{\Delta t}{2} \partial_\beta^{(1)} \left(1 - \frac{1}{2\tau_{s,1}}\right) \partial_\gamma^{(1)} \Phi_{\alpha\beta\gamma} = 0. \label{eq:CE2_mom}
	\end{gather}
\end{subequations}
Substituting Eq.~\eqref{eq:A1} into Eq.~\eqref{eq:CE2_mom} and simplifying using the leading-order conservation laws \eqref{eq:CE1}, we obtain:
\begin{equation}\label{eq:CE2_simplified}
	\partial_t^{(2)} (\rho u_\alpha) = \partial_\beta^{(1)} \left(\tau_{s,1} - \frac{1}{2}\right) \Delta t \left[ \rho c_s^2 (\partial_\alpha^{(1)} u_\beta + \partial_\beta^{(1)} u_\alpha) + (F_\alpha^{(1)} u_\beta + u_\alpha F_\beta^{(1)}) \right].
\end{equation}

\paragraph{Recovered macroscopic equations.}
Combining Eqs.~\eqref{eq:CE1_mass}$\times\varepsilon$ + \eqref{eq:CE2_mass}$\times\varepsilon^2$ and Eqs.~\eqref{eq:CE1_mom}$\times\varepsilon$ + \eqref{eq:CE2_simplified}$\times\varepsilon^2$, we recover:
\begin{subequations}\label{eq:recovered_NSE}
	\begin{gather}
		\partial_t \rho + \partial_\alpha (\rho u_\alpha) = 0, \\
		\partial_t (\rho u_\alpha) + \partial_\beta (\rho u_\alpha u_\beta + \rho c_s^2 \delta_{\alpha\beta}) = \partial_\beta \left[ \mu (\partial_\alpha u_\beta + \partial_\beta u_\alpha) \right] + F_\alpha + \partial_\beta \left[ \mu (F_\alpha u_\beta + u_\alpha F_\beta) / (\rho c_s^2) \right],
	\end{gather}
\end{subequations}
where the dynamic viscosity is:
\begin{equation}\label{eq:viscosity}
	\mu = \rho c_s^2 \left(\tau_{s,1} - \frac{1}{2}\right) \Delta t.
\end{equation}
The last term on the right-hand side of the momentum equation is an artifact of the simplified forcing scheme (Scheme~II). As noted by Guo et al.\ \cite{guo2002discrete}, this term does not compromise accuracy for low-Mach-number flows, and Postma et al.\ \cite{postma2020force} confirmed its suitability for two-relaxation-time models. When Scheme~I (the Guo forcing) is employed, this artifact vanishes and the standard Navier--Stokes equations are recovered exactly.

\paragraph{Role of the free parameter $\tau_{s,2}$}
A critical observation from the above analysis is that the free relaxation parameter $\tau_{s,2}$ is entirely absent from the recovered macroscopic equations \eqref{eq:recovered_NSE} and the viscosity relation \eqref{eq:viscosity}. This confirms that $\tau_{s,2}$ does not affect the physical transport coefficients. Its role is confined to the non-hydrodynamic level: it controls the damping rate of the third-order (antisymmetric) non-equilibrium modes, which are filtered out during the regularization process. The P-TRT framework makes this role transparent---$\tau_{s,2}$ governs how quickly the antisymmetric non-equilibrium component, purified of any first-order residual via the force correction \eqref{eq:R3_force}, is relaxed toward zero.

\paragraph{Implications for numerical stability.}
Although the Chapman--Enskog analysis confirms macroscopic consistency, it is inherently asymptotic and does not directly address numerical stability, which depends on the spectral properties of the fully discrete evolution operator. To bridge this gap, a linear stability analysis is presented in the next subsection, where the eigenvalue structure of the P-TRT collision operator is compared with that of the standard TRT model.

\subsection{Linear Stability Analysis}
\label{sec:LSA}

The Chapman--Enskog expansion establishes the macroscopic consistency of the P-TRT model but provides no direct information about numerical stability, which is governed by the spectral radius of the fully discrete evolution operator. In this subsection, we perform a linear stability analysis (LSA) in the moment space to rigorously characterize how the ghost-mode annihilation in the P-TRT collision operator affects the eigenvalue spectrum of the scheme, in comparison with the standard TRT model.

\subsubsection{Linearization and Fourier Analysis}

Consider a uniform rest state $f_i^{(0)} = w_i \rho_0$, $\boldsymbol{u}^{(0)} = \boldsymbol{0}$. A small perturbation is introduced:
\begin{equation}\label{eq:lsa_perturb}
	f_i(\boldsymbol{x}, t) = f_i^{(0)} + \epsilon \, \delta f_i(\boldsymbol{x}, t), \quad |\epsilon| \ll 1,
\end{equation}
where the perturbation takes the Fourier normal-mode form:
\begin{equation}\label{eq:fourier_mode}
	\delta f_i(\boldsymbol{x}, t) = \hat{f}_i \, z^n \, e^{i \boldsymbol{k} \cdot \boldsymbol{x}},
\end{equation}
with $\boldsymbol{k} = (k_x, k_y)$ the wave vector, $n$ the discrete time index, and $z \in \mathbb{C}$ the temporal amplification factor. The numerical scheme is linearly stable if and only if $|z| \le 1$ for all $\boldsymbol{k} \in [-\pi/\Delta x, \, \pi/\Delta x]^2$.

One complete time step of the lattice Boltzmann algorithm consists of collision followed by streaming. In Fourier space, the streaming step is represented by the diagonal phase-shift matrix
\begin{equation}\label{eq:streaming_matrix}
	\boldsymbol{A} = \mathrm{diag}\!\left(e^{-i \boldsymbol{k} \cdot \boldsymbol{e}_i}\right)_{i=0}^{Q-1},
\end{equation}
and the linearized collision step is represented by a collision matrix $\boldsymbol{C}$ acting on the perturbation vector $\delta\boldsymbol{f} = (\delta f_0, \ldots, \delta f_{Q-1})^T$. The full evolution is governed by
\begin{equation}\label{eq:evolution_f}
	\delta\boldsymbol{f}^{\,n+1} = \boldsymbol{L} \, \delta\boldsymbol{f}^{\,n}, \qquad \boldsymbol{L} = \boldsymbol{A} \, \boldsymbol{C},
\end{equation}
and the stability is determined by the eigenvalues of the evolution operator $\boldsymbol{L}$.

\subsubsection{Moment-Space Representation}

The analysis is most transparent in the moment space. Define the moment vector
\begin{equation}\label{eq:moment_vector}
	\delta\boldsymbol{m} = \boldsymbol{M} \, \delta\boldsymbol{f},
\end{equation}
where $\boldsymbol{M}$ is the transformation matrix whose rows correspond to the orthogonal moment basis vectors identified in Section~\ref{sec:decomposition}. For the D2Q9 lattice, the nine moments are organized into four groups:
\begin{equation}\label{eq:moment_groups}
	\delta\boldsymbol{m} = \begin{pmatrix} \delta\boldsymbol{m}_{\mathrm{cons}} \\ \delta\boldsymbol{m}_{(2)} \\ \delta\boldsymbol{m}_{(3)} \\ \delta m_G \end{pmatrix},
\end{equation}
where $\delta\boldsymbol{m}_{\mathrm{cons}} = (\delta\rho,\, \delta j_x,\, \delta j_y)^T$ contains the conserved quantities (3 modes), $\delta\boldsymbol{m}_{(2)} = (\delta\pi_{xx},\, \delta\pi_{yy},\, \delta\pi_{xy})^T$ the second-order stress modes (3 modes), $\delta\boldsymbol{m}_{(3)} = (\delta q_x,\, \delta q_y)^T$ the third-order energy-flux modes (2 modes), and $\delta m_G$ the ghost mode (1 mode). The first eight modes constitute the \emph{hydrodynamic sector} $H$, and the last mode the \emph{ghost sector} $G$.

In this basis, the linearized collision matrices of the two models take diagonal form. For the P-TRT model:
\begin{equation}\label{eq:collision_PTRT}
	\tilde{\boldsymbol{C}}_{\mathrm{P\text{-}TRT}} = \mathrm{diag}\!\left(\underbrace{1,\, 1,\, 1}_{\mathrm{cons.}},\; \underbrace{1-\frac{1}{\tau_{s,1}},\, 1-\frac{1}{\tau_{s,1}},\, 1-\frac{1}{\tau_{s,1}}}_{2\text{nd order}},\; \underbrace{1-\frac{1}{\tau_{s,2}},\, 1-\frac{1}{\tau_{s,2}}}_{3\text{rd order}},\; \underbrace{0}_{\text{ghost}}\right),
\end{equation}
and for the standard TRT model:
\begin{equation}\label{eq:collision_TRT}
	\tilde{\boldsymbol{C}}_{\mathrm{TRT}} = \mathrm{diag}\!\left(\underbrace{1,\, 1,\, 1}_{\mathrm{cons.}},\; \underbrace{1-\frac{1}{\tau_{s,1}},\, 1-\frac{1}{\tau_{s,1}},\, 1-\frac{1}{\tau_{s,1}}}_{2\text{nd order}},\; \underbrace{1-\frac{1}{\tau_{s,2}},\, 1-\frac{1}{\tau_{s,2}}}_{3\text{rd order}},\; \underbrace{1-\frac{1}{\tau_{s,1}}}_{\text{ghost}}\right).
\end{equation}
The sole difference is the treatment of the ghost mode: P-TRT sets its relaxation factor to zero (complete annihilation), whereas the standard TRT retains the factor $1 - 1/\tau_{s,1}$, identical to the stress modes. This can be equivalently understood as P-TRT applying an effective relaxation rate of $1/\tau_G = 1$ (instantaneous equilibration) to the ghost subspace, while TRT applies $1/\tau_G = 1/\tau_{s,1}$.

\subsubsection{Block Structure of the Evolution Operator}

The moment-space evolution operator is $\tilde{\boldsymbol{L}} = \tilde{\boldsymbol{A}} \, \tilde{\boldsymbol{C}}$, where $\tilde{\boldsymbol{A}} = \boldsymbol{M} \boldsymbol{A} \boldsymbol{M}^{-1}$ is the moment-space streaming matrix. The streaming matrix $\tilde{\boldsymbol{A}}$ is generally \emph{not} diagonal in moment space; it couples different moments through the advection step. Writing $\tilde{\boldsymbol{A}}$ in a $2 \times 2$ block form partitioned by the hydrodynamic ($H$, 8 modes) and ghost ($G$, 1 mode) sectors:
\begin{equation}\label{eq:streaming_block}
	\tilde{\boldsymbol{A}} = \begin{pmatrix} \tilde{\boldsymbol{A}}_{HH} & \tilde{\boldsymbol{A}}_{HG} \\ \tilde{\boldsymbol{A}}_{GH} & \tilde{A}_{GG} \end{pmatrix},
\end{equation}
the collision matrices~\eqref{eq:collision_PTRT}--\eqref{eq:collision_TRT} have the corresponding block forms:
\begin{equation}\label{eq:collision_block}
	\tilde{\boldsymbol{C}}_{\mathrm{P\text{-}TRT}} = \begin{pmatrix} \boldsymbol{C}_{HH} & \boldsymbol{0} \\ \boldsymbol{0} & 0 \end{pmatrix}, \qquad
	\tilde{\boldsymbol{C}}_{\mathrm{TRT}} = \begin{pmatrix} \boldsymbol{C}_{HH} & \boldsymbol{0} \\ \boldsymbol{0} & 1 - 1/\tau_{s,1} \end{pmatrix},
\end{equation}
where $\boldsymbol{C}_{HH}$ is the $8 \times 8$ hydrodynamic collision block, which is \emph{identical} in both models. The zero column in the P-TRT collision matrix ensures that the ghost mode contributes nothing to the post-collision state, regardless of its pre-collision value.

Multiplying out, the evolution operators become:
\begin{equation}\label{eq:L_PTRT}
	\tilde{\boldsymbol{L}}_{\mathrm{P\text{-}TRT}} = \begin{pmatrix} \tilde{\boldsymbol{A}}_{HH} \boldsymbol{C}_{HH} & \boldsymbol{0} \\ \tilde{\boldsymbol{A}}_{GH} \boldsymbol{C}_{HH} & 0 \end{pmatrix},
\end{equation}
\begin{equation}\label{eq:L_TRT}
	\tilde{\boldsymbol{L}}_{\mathrm{TRT}} = \begin{pmatrix} \tilde{\boldsymbol{A}}_{HH} \boldsymbol{C}_{HH} & (1-1/\tau_{s,1})\,\tilde{\boldsymbol{A}}_{HG} \\ \tilde{\boldsymbol{A}}_{GH} \boldsymbol{C}_{HH} & (1-1/\tau_{s,1})\,\tilde{A}_{GG} \end{pmatrix}.
\end{equation}
Two structural differences are immediately apparent. First, $\tilde{\boldsymbol{L}}_{\mathrm{P\text{-}TRT}}$ is \emph{block lower-triangular}, so its eigenvalues decouple into those of $\tilde{\boldsymbol{A}}_{HH} \boldsymbol{C}_{HH}$ and the scalar zero entry. Second, $\tilde{\boldsymbol{L}}_{\mathrm{TRT}}$ has a non-zero upper-right block $(1-1/\tau_{s,1})\,\tilde{\boldsymbol{A}}_{HG}$, which couples the ghost mode back into the hydrodynamic sector during streaming, and a non-zero ghost diagonal entry $(1-1/\tau_{s,1})\,\tilde{A}_{GG}$.

\subsubsection{Eigenvalue Structure and Stability Decoupling}

The block lower-triangular structure of Eq.~\eqref{eq:L_PTRT} leads directly to the following result.

\begin{theorem}[Ghost-mode annihilation and stability decoupling]\label{thm:stability}
	For the P-TRT model, the eigenvalues of the evolution operator $\tilde{\boldsymbol{L}}_{\mathrm{P\text{-}TRT}}$ consist of:
	\begin{enumerate}
		\item[(i)] The ghost eigenvalue $z_G = 0$, which is unconditionally stable.
		\item[(ii)] Eight hydrodynamic eigenvalues $\{z_\alpha\}_{\alpha=1}^{8}$, determined by
		\begin{equation}\label{eq:hydro_eigen}
			\det\!\left(\tilde{\boldsymbol{A}}_{HH} \boldsymbol{C}_{HH} - z\,\boldsymbol{I}_8\right) = 0.
		\end{equation}
	\end{enumerate}
	Consequently, the numerical stability of the P-TRT scheme is governed \emph{exclusively} by the hydrodynamic eigenvalues, and is completely independent of the ghost mode.
\end{theorem}

\begin{proof}
	The characteristic polynomial of the block lower-triangular matrix~\eqref{eq:L_PTRT} factors as:
	\begin{equation}
		\det\!\left(\tilde{\boldsymbol{L}}_{\mathrm{P\text{-}TRT}} - z\,\boldsymbol{I}_9\right) = \det\!\left(\tilde{\boldsymbol{A}}_{HH} \boldsymbol{C}_{HH} - z\,\boldsymbol{I}_8\right) \cdot (-z) = 0.
	\end{equation}
	The factor $(-z) = 0$ yields $z_G = 0$, with $|z_G| = 0 < 1$.  \qed
\end{proof}

In contrast, for the standard TRT model, the ghost eigenvalue is generally non-zero. When the coupling between the ghost and hydrodynamic sectors via streaming is weak (i.e., when the off-diagonal blocks $\tilde{\boldsymbol{A}}_{HG}$ and $\tilde{\boldsymbol{A}}_{GH}$ are small relative to the diagonal blocks), the ghost eigenvalue can be approximated as:
\begin{equation}\label{eq:ghost_eigen_TRT}
	z_G^{\mathrm{TRT}}(\boldsymbol{k}) \approx \left(1 - \frac{1}{\tau_{s,1}}\right) \tilde{A}_{GG}(\boldsymbol{k}),
\end{equation}
where $\tilde{A}_{GG}(\boldsymbol{k})$ is the ghost-sector entry of the streaming matrix in moment space. At high Reynolds numbers, $\tau_{s,1} \to 1/2$ (i.e., $1/\tau_{s,1} \to 2$), so the pre-factor $|1 - 1/\tau_{s,1}| \to 1$, and $|z_G^{\mathrm{TRT}}|$ approaches $|\tilde{A}_{GG}(\boldsymbol{k})|$. For certain high-wavenumber modes, $|\tilde{A}_{GG}|$ can exceed unity, causing the ghost eigenvalue to escape the unit circle and triggering numerical instability.

These observations lead to the following comparison.

\begin{remark}[Spectral radius comparison]\label{rmk:spectral}
	Let $\rho(\cdot)$ denote the spectral radius. For any fixed parameter set $(\tau_{s,1}, \tau_{s,2})$ and wave vector $\boldsymbol{k}$:
	\begin{equation}\label{eq:spectral_comparison}
		\rho\!\left(\tilde{\boldsymbol{L}}_{\mathrm{P\text{-}TRT}}\right) = \max_{\alpha \in H} |z_\alpha| \;\le\; \max\!\left\{\max_{\alpha \in H} |z_\alpha|,\; |z_G^{\mathrm{TRT}}|\right\} = \rho\!\left(\tilde{\boldsymbol{L}}_{\mathrm{TRT}}\right).
	\end{equation}
	Equality holds if and only if the ghost eigenvalue of TRT does not exceed the largest hydrodynamic eigenvalue. When the stability bottleneck of the standard TRT lies in the ghost mode (i.e., $|z_G^{\mathrm{TRT}}| > \max_{\alpha \in H}|z_\alpha|$), the P-TRT model is \emph{strictly} more stable, permitting a larger critical relaxation rate and hence a lower attainable viscosity.
\end{remark}

Combining Theorem~\ref{thm:stability} with the equivalence results of Theorems~\ref{thm:2nd} and~\ref{thm:3rd}, the same spectral properties apply to the original TRT-RLB model, since the two formulations produce the identical post-collision state for D2Q9 and D3Q19 lattices.

\begin{remark}[On the optimal value of $\tau_{s,2}$]\label{rmk:optimal_tau}
	With the ghost mode annihilated, the stability of the P-TRT scheme depends on the eight hydrodynamic eigenvalues $\{z_\alpha(\boldsymbol{k},\, \tau_{s,1},\, \tau_{s,2})\}$, whose moduli depend on $\tau_{s,2}$ through the collision factor $(1 - 1/\tau_{s,2})$ applied to the third-order modes. The optimal value of $\tau_{s,2}$ corresponds to a minimax condition:
	\begin{equation}\label{eq:minimax}
		\tau_{s,2}^{\mathrm{opt}} = \arg\min_{\tau_{s,2}} \; \max_{\boldsymbol{k},\, \alpha \in H} |z_\alpha(\boldsymbol{k},\, \tau_{s,1},\, \tau_{s,2})|.
	\end{equation}
	The determination of this optimum requires a parametric numerical eigenvalue computation of the $8 \times 8$ hydrodynamic evolution matrix over the full Brillouin zone $\boldsymbol{k} \in [-\pi/\Delta x,\, \pi/\Delta x]^2$, which constitutes a substantial numerical algebra problem in its own right. Such a complete parametric study is beyond the scope of the present work and is left for future investigation. Nevertheless, the empirical observation of an optimal stability plateau near $1/\tau_{s,2} \approx 1.6$ reported in Section~\ref{sec:3} is consistent with this minimax interpretation. We note that the analogous problem for the classical TRT framework---the analytical determination of the optimal ``magic parameter'' $\Lambda$---remains similarly open in the existing literature, where optimal values are typically identified through numerical experimentation \cite{TRT2008study,TRT2008two}.
\end{remark}

\section{Numerical Results and Discussion}
\label{sec:3}

In this section, we perform a series of benchmark simulations to validate the P-TRT model from three complementary perspectives: computational efficiency, numerical stability, and physical accuracy. The central expectation, derived from the mathematical equivalence established in Section~\ref{sec:equivalence}, is that the P-TRT model should produce results identical to the original TRT-RLB model while requiring significantly less computational effort. Four benchmark problems are considered: the double shear layer flow (stability and efficiency), the 2D decaying Taylor--Green vortex (accuracy and convergence), the force-driven Poiseuille flow (boundary properties), and the creeping flow past a square cylinder (robustness at ultra-low Reynolds numbers). For quantitative assessment, the relative $L_2$ error norm is defined as:
\begin{equation}
	L_2 = \sqrt{\frac{\sum_{\mathbf{x}} (q_n(\mathbf{x}, t) - q_a(\mathbf{x}, t))^2}{\sum_{\mathbf{x}} q_a^2(\mathbf{x}, t)}},
\end{equation}
where $q_a(\mathbf{x}, t)$ and $q_n(\mathbf{x}, t)$ represent the analytical and numerical solutions, respectively. For steady-state flows, the convergence criterion is:
\begin{equation}
	\max \left( \left| \frac{u_\alpha(\mathbf{x}, t + 10000\Delta t) - u_\alpha(\mathbf{x}, t)}{u_c} \right| \right) < 10^{-7}.
\end{equation}

\subsection{Double Shear Layer: Stability and Efficiency}

The double shear layer flow is selected to benchmark numerical stability at high Reynolds numbers. Its periodic boundary conditions also enable precise quantification of computational efficiency by eliminating boundary-handling overheads. The computational domain is a two-dimensional square region with periodic boundary conditions on all sides, defined as $(x, y) \in [0, L]^2$ \cite{WOS:001361481500001}. The initial velocity field consists of two longitudinal shear layers and a transverse perturbation, given by:
\begin{equation}
	u_x(\mathbf{x}, 0) =
	\begin{cases}
		u_c \tanh \left[ \kappa \left( \dfrac{y}{L} - \dfrac{1}{4} \right) \right], & \dfrac{y}{L} \le \dfrac{1}{2}, \\[6pt]
		u_c \tanh \left[ \kappa \left( \dfrac{3}{4} - \dfrac{y}{L} \right) \right], & \dfrac{y}{L} > \dfrac{1}{2},
	\end{cases}
\end{equation}
and
\begin{equation}
	u_y(\mathbf{x}, 0) = u_c \delta \sin \left[ 2\pi \left( \frac{x}{L} + \frac{1}{4} \right) \right].
\end{equation}
Here, $u_c$ represents the characteristic velocity, $\kappa$ controls the thickness of the shear layers, and $\delta$ determines the magnitude of the initial perturbation.

In our simulations, the parameters are set as follows: $\kappa = 80$, $\delta = 0.05$, $L = 1$, $c = 1$, and $\rho_0 = 1$. A uniform grid with resolution $N \times N = 128^2$ is used, giving $\Delta x = L/N$. The flow is characterized by the Reynolds number and Mach number:
\begin{equation}
	Re = \frac{\rho_0 u_c L}{\mu}, \quad Ma = \frac{u_c}{c_s},
\end{equation}
and the characteristic time is $t_c = L/u_c$.

\subsubsection{Stability at Super-High Reynolds Numbers}

To quantify the impact of the free relaxation parameter $\tau_{s,2}$ on numerical stability, we measured the maximum critical Mach number ($Ma_c$). A simulation is considered stable if it runs to $t = 2t_c$ without divergence. The precision for $Ma_c$ is set to $\Delta Ma = 0.01$. We investigated stability across a wide range of Reynolds numbers, with the inverse relaxation parameter $1/\tau_{s,2}$ varying from 0.1 to 2.0. For a comprehensive comparison, the P-TRT model is tested alongside the BGK, TRT, RLB, and original TRT-RLB models \cite{WOS:001361481500001} under identical conditions.

The stability comparison is presented in Figure~\ref{fig:ma}. The SRT-type models (BGK and RLB) are represented by horizontal lines without markers, indicating that their stability is independent of $\tau_{s,2}$. When $Re \ge 5 \times 10^4$, both the standard BGK and TRT models exhibit divergence, resulting in $Ma_c = 0$. In contrast, the P-TRT model consistently demonstrates superior numerical stability. Notably, the stability profile of the P-TRT model is virtually identical to that of the original TRT-RLB model across the entire parameter space. This numerical identity is a direct manifestation of the mathematical equivalence established in Theorems~\ref{thm:2nd} and~\ref{thm:3rd}: since the two models produce the identical post-collision distribution for the D2Q9 lattice, they must exhibit the same stability boundary.

The stability reaches its optimal state near $1/\tau_{s,2} \approx 1.6$. Even at super-high Reynolds numbers ranging from $5 \times 10^4$ to $1 \times 10^7$, the maximum $Ma_c$ remains robustly above 0.51. These observations are fully consistent with the linear stability analysis of Section~\ref{sec:LSA}. According to Theorem~\ref{thm:stability}, the P-TRT collision operator annihilates the ghost eigenvalue ($z_G = 0$), so the scheme's stability is governed exclusively by the eight hydrodynamic eigenvalues. In contrast, the standard TRT model relaxes the ghost mode at the same rate as the stress modes (both governed by $\tau_{s,1}$); as $\tau_{s,1} \to 1/2$ at high Reynolds numbers, the ghost amplification factor $|z_G^{\mathrm{TRT}}|$ can approach or exceed unity, triggering instability (cf.\ Remark~\ref{rmk:spectral}). The P-TRT model eliminates this failure mechanism entirely by removing the ghost mode from the collision step.

The parameter $\tau_{s,2}$, which controls the antisymmetric (third-order) relaxation, provides an additional degree of freedom for optimizing the spectral properties of the remaining hydrodynamic modes. The observed optimum near $1/\tau_{s,2} \approx 1.6$ reflects the minimax balance described in Remark~\ref{rmk:optimal_tau}: at this value, the maximum modulus of the hydrodynamic eigenvalues across the Brillouin zone is minimized. A precise analytical determination of this optimum would require a parametric eigenvalue study of the $8 \times 8$ hydrodynamic evolution matrix, which is beyond the scope of the present work. Nevertheless, the robust empirical plateau observed across six decades of Reynolds number (Figure~\ref{fig:ma}) confirms that this optimal value is insensitive to $\tau_{s,1}$, consistent with the theoretical prediction that the ghost-mode decoupling renders the stability landscape dependent only on the hydrodynamic sector.

\begin{figure}[H]
	\centering
	\begin{subfigure}[b]{0.3\textwidth}
		\centering
		\includegraphics[width=\linewidth]{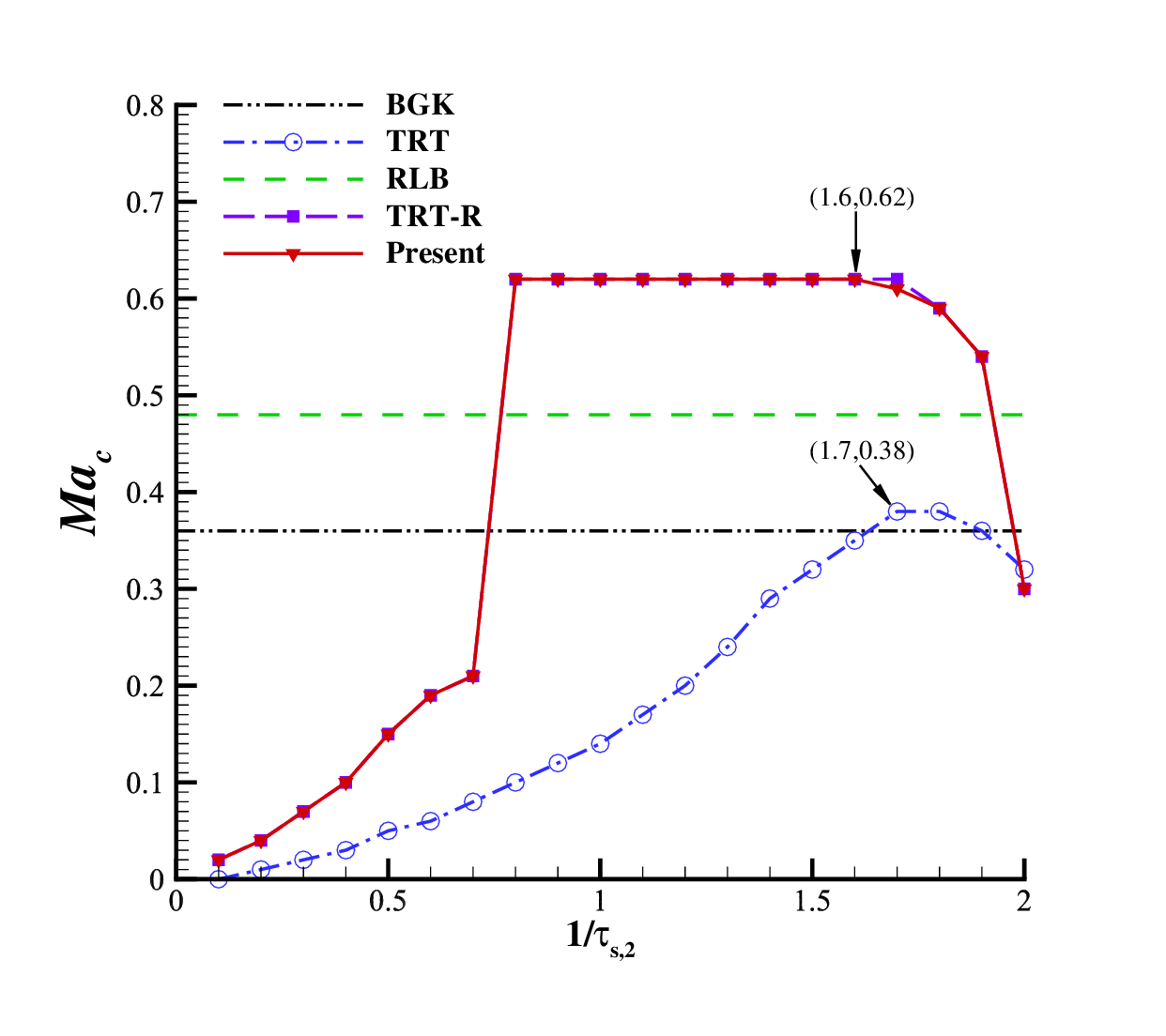}
		\caption{$Re = 5 \times 10^3$}
		\label{subfig:a}
	\end{subfigure}
	\hfill
	\begin{subfigure}[b]{0.3\textwidth}
		\centering
		\includegraphics[width=\linewidth]{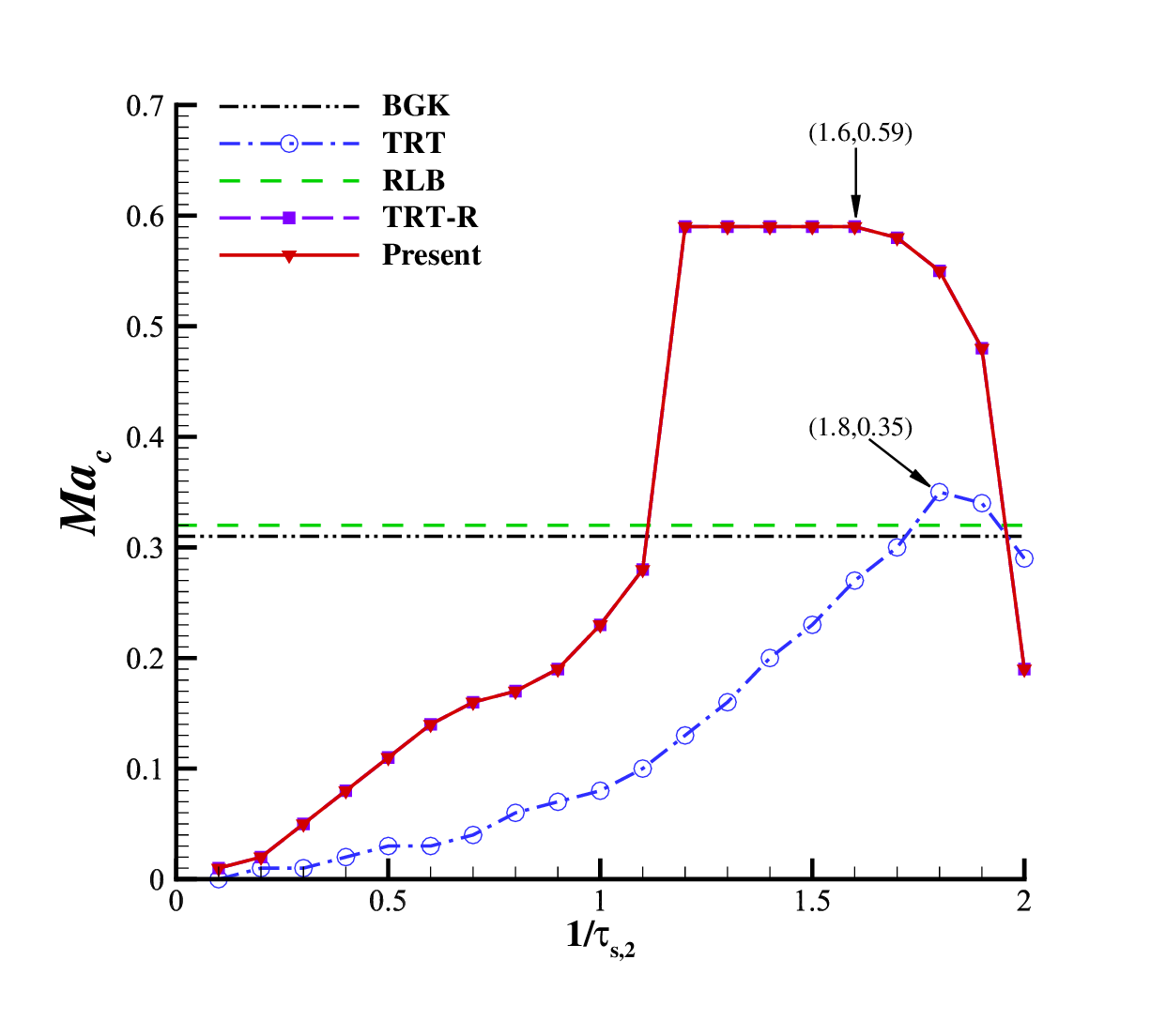}
		\caption{$Re = 1 \times 10^4$}
		\label{subfig:b}
	\end{subfigure}
	\hfill
	\begin{subfigure}[b]{0.3\textwidth}
		\centering
		\includegraphics[width=\linewidth]{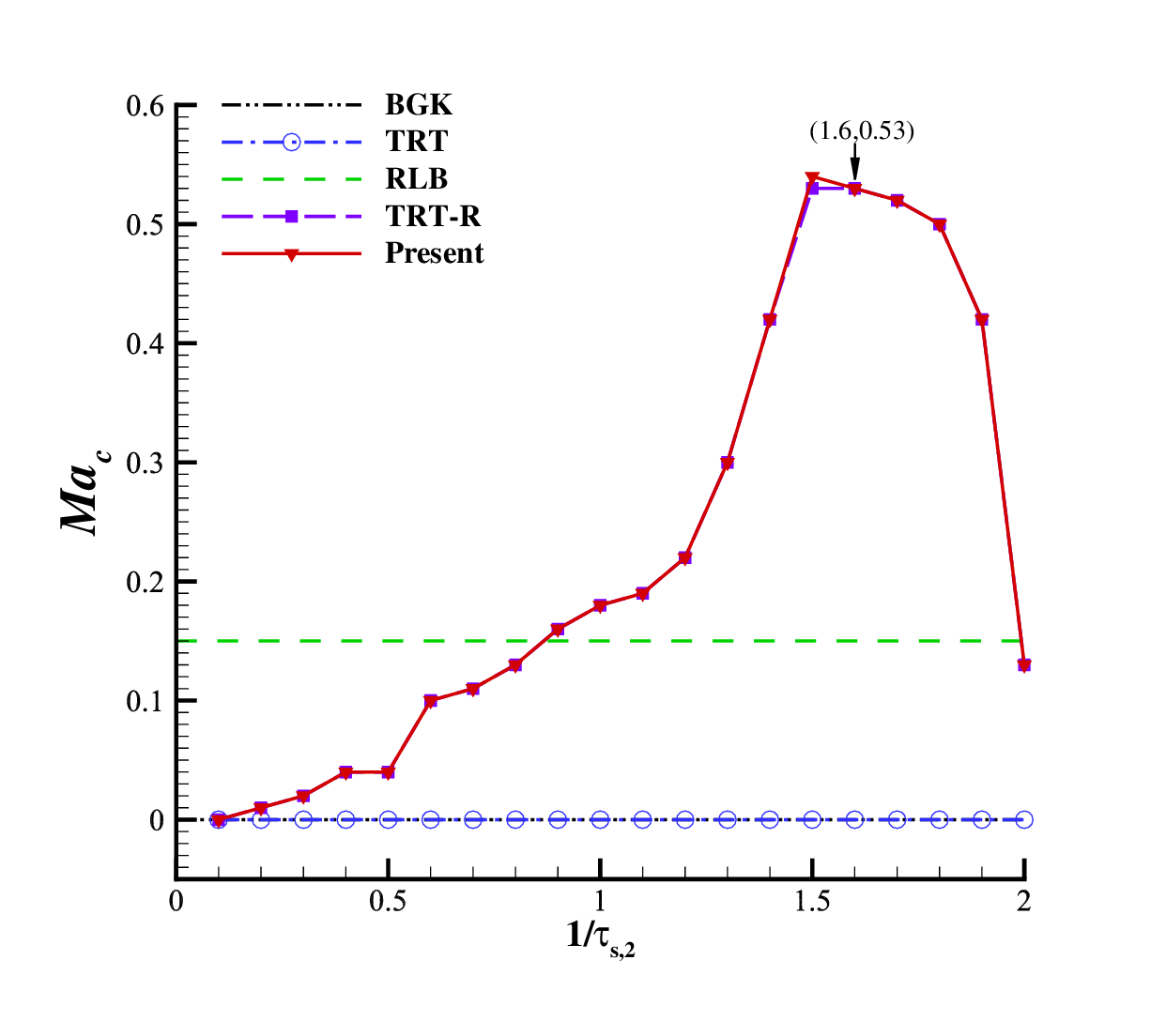}
		\caption{$Re = 5 \times 10^4$}
		\label{subfig:c}
	\end{subfigure}
	
	\vspace{0.5em}
	\begin{subfigure}[b]{0.3\textwidth}
		\centering
		\includegraphics[width=\linewidth]{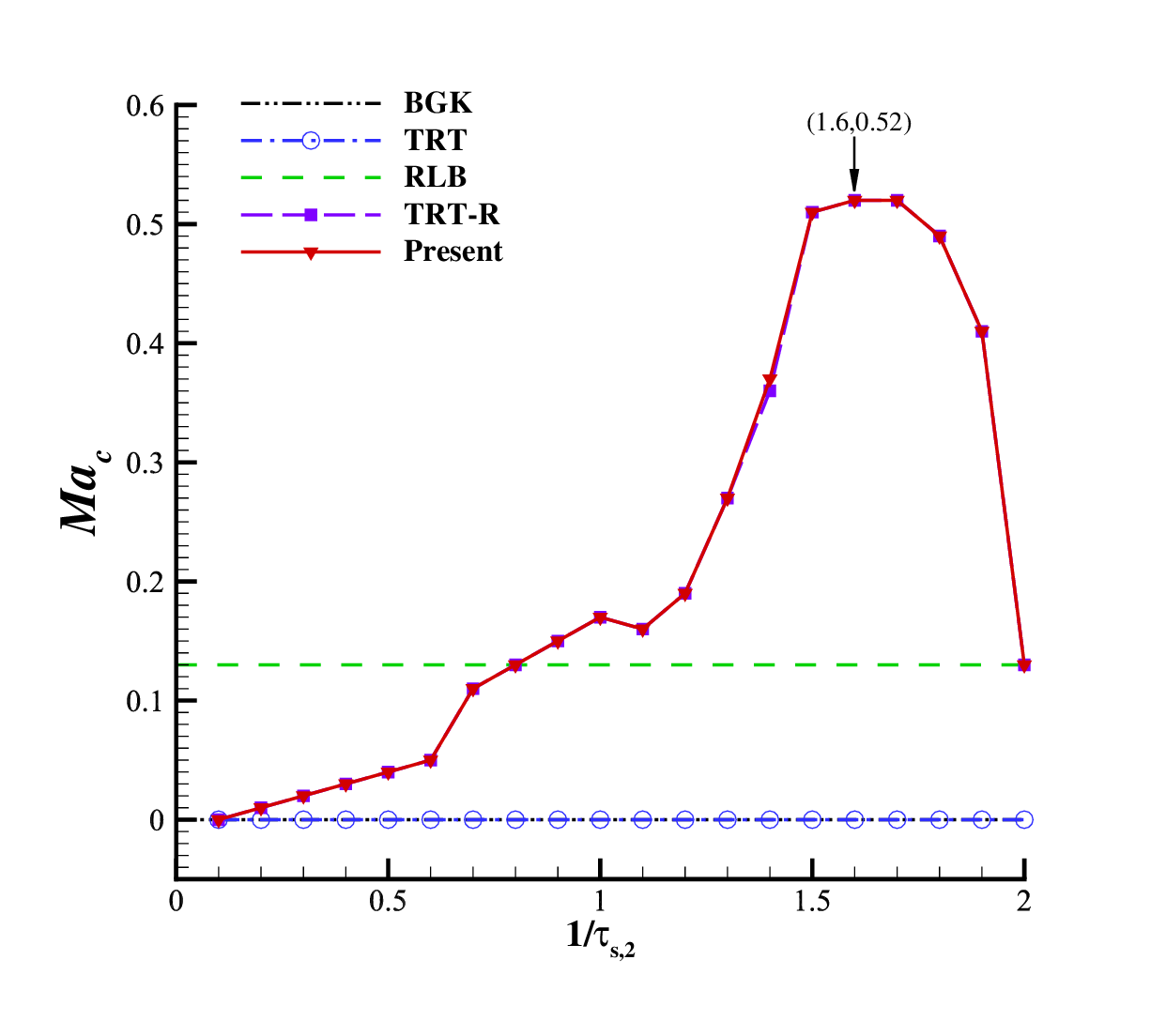}
		\caption{$Re = 1 \times 10^5$}
		\label{subfig:d}
	\end{subfigure}
	\hfill
	\begin{subfigure}[b]{0.3\textwidth}
		\centering
		\includegraphics[width=\linewidth]{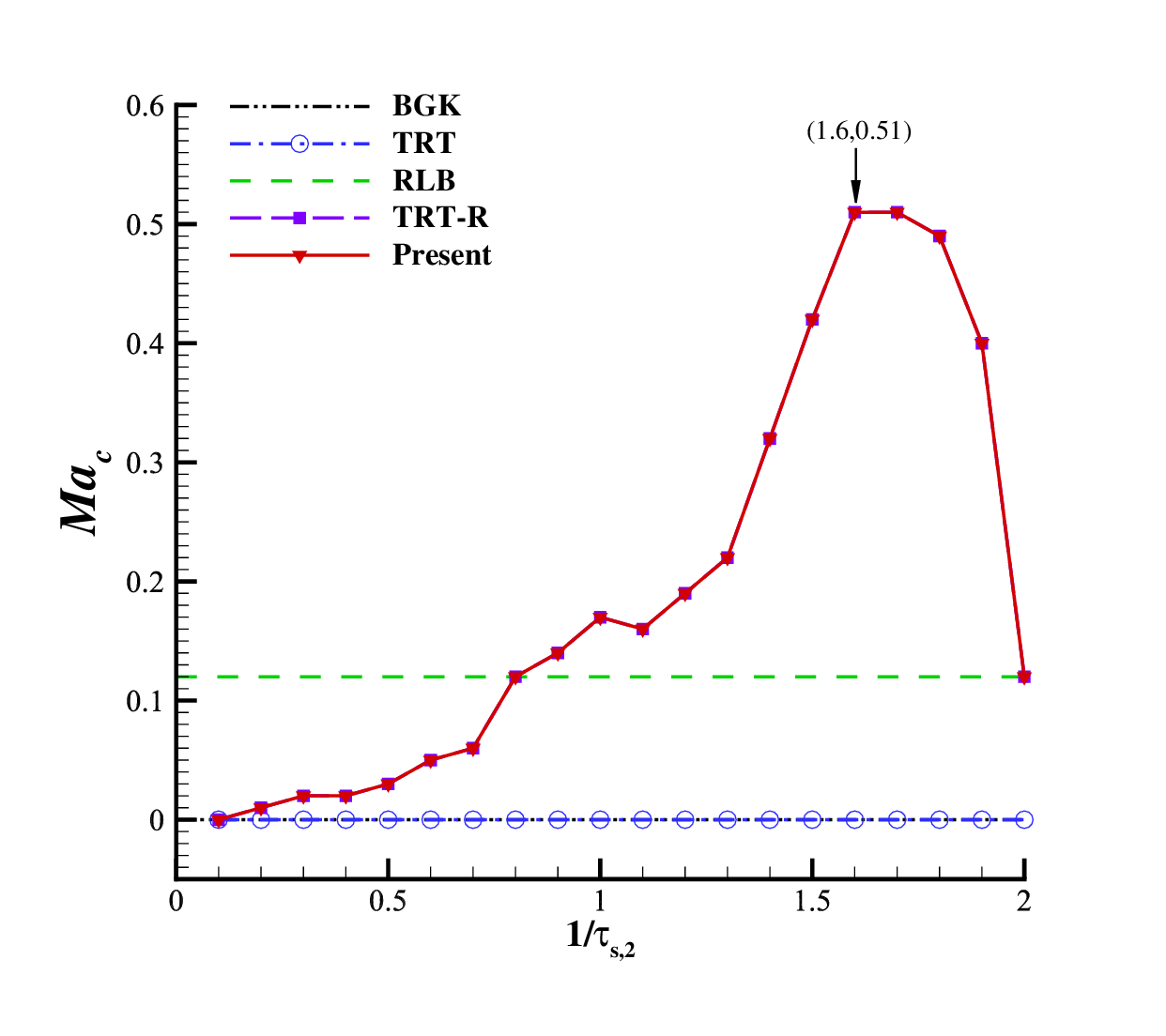}
		\caption{$Re = 1 \times 10^6$}
		\label{subfig:e}
	\end{subfigure}
	\hfill
	\begin{subfigure}[b]{0.3\textwidth}
		\centering
		\includegraphics[width=\linewidth]{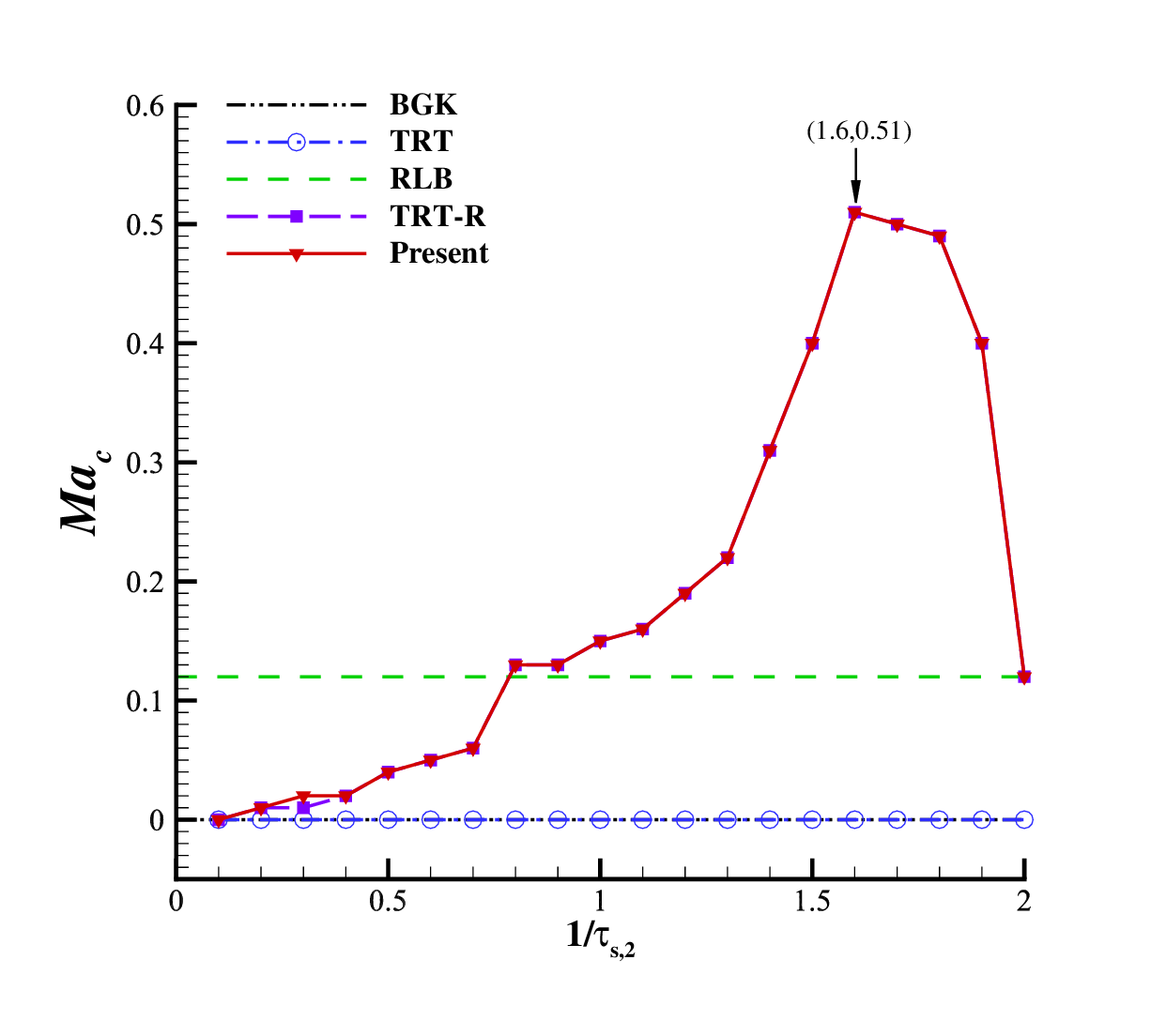}
		\caption{$Re = 1 \times 10^7$}
		\label{subfig:f}
	\end{subfigure}
	
	\caption{Maximum critical Mach number $Ma_c$ versus the inverse relaxation parameter $1/\tau_{s,2}$ in the double shear layer simulations at various Reynolds numbers. The P-TRT model (labeled ``Present'') demonstrates robust stability at super-high Reynolds numbers where standard BGK and TRT models fail to converge. The stability profile is identical to the original TRT-RLB model, confirming the mathematical equivalence.}
	\label{fig:ma}
\end{figure}

\subsubsection{Computational Efficiency Evaluation}\label{sec3.1.2}

This test case is also employed to benchmark the computational efficiency of the P-TRT model. The periodic boundaries remove boundary-related computational costs, thereby focusing the assessment on the collision operator. We conducted simulations across a range of Reynolds numbers from $10$ to $10^4$. To ensure a fair comparison, all simulations were performed on the same computing platform, and all runs were terminated only upon satisfying the strict convergence criterion. The speedup ratio is defined as:
\begin{equation}
	S_p = \frac{T_{\text{TRT-RLB}}}{T_{\text{P-TRT}}}.
\end{equation}
The comparative results are summarized in Table~\ref{tab:efficiency}. The P-TRT model consistently outperforms the original TRT-RLB framework. The relative reduction in computational cost ranges from 29\% to 35\%, yielding an average speedup ratio of approximately 1.45. This measured speedup is consistent with the theoretical FLOP reduction predicted in Table~\ref{tab:flops_total}: for the D2Q9 lattice, the non-equilibrium collision cost is reduced by 71\%, and the observed wall-clock reduction of approximately 30\% confirms that the collision step constitutes a substantial fraction of the total computational work for this periodic benchmark.

\begin{table}[h]
	\centering
	\caption{Comparison of execution time (in seconds) and computational efficiency between the original TRT-RLB and the P-TRT models for the double shear layer simulation across a wide range of Reynolds numbers.}
	\label{tab:efficiency}
	\begin{tabular}{ccccc}
		\hline
		$Re$ & TRT-RLB (s) & P-TRT (s) & Reduction (\%) & Speedup \\
		\hline
		10      & 6.93      & 4.53      & 34.63 & 1.53 \\
		100     & 43.40     & 30.75     & 29.14 & 1.41 \\
		1,000   & 328.34    & 232.14    & 29.30 & 1.41 \\
		10,000  & 2,503.53  & 1747.91   & 30.18 & 1.43 \\
		\hline
	\end{tabular}
\end{table}

\subsection{2D Decaying Taylor--Green Vortex: Accuracy and Convergence}

To rigorously validate the accuracy and grid convergence of the P-TRT model, we simulate the 2D decaying Taylor--Green vortex flow. This canonical benchmark admits an exact analytical solution, enabling precise quantification of numerical errors. The flow is defined within a square domain $(x, y) \in [0, 1]^2$ with periodic boundary conditions. The analytical solutions for the velocity field and density are given by \cite{kruger2017lattice,pearson1965computational}:
\begin{subequations}\label{eq:TG}
	\begin{align}
		u_x(\mathbf{x}, t) &= -u_0 \cos(2\pi x) \sin(2\pi y) \, e^{-t/t_d}, \\
		u_y(\mathbf{x}, t) &= u_0 \cos(2\pi y) \sin(2\pi x) \, e^{-t/t_d}, \\
		\rho(\mathbf{x}, t) &= \rho_0 - \rho_0 \frac{u_0^2}{4c_s^2} \left[\cos(4\pi x) + \cos(4\pi y)\right] e^{-t/t_d},
	\end{align}
\end{subequations}
where $t_d = 1/(8\nu\pi^2)$ is the vortex decay time.

In our simulations, the physical parameters are set to $u_0 = 0.01$, $\rho_0 = 1$, and $\nu = 0.01$, with $L = 1$. The time step is determined by the scaling relation $\Delta x^2/\Delta t = 0.01\pi^2$. The magic parameter is defined as:
\begin{equation}
	\Lambda_s = \left(\tau_{s,1} - \frac{1}{2}\right)\left(\tau_{s,2} - \frac{1}{2}\right),
\end{equation}
and is set to $\Lambda_s = 1/4$ for this test case.

To assess accuracy, a simulation was performed with grid resolution $N = 64$. Figure~\ref{fig:tg_solution} compares the computed velocity profiles $u_x(L/2, y)$ and $u_y(x, L/2)$ at various time instants with the analytical solutions. The results show excellent agreement, confirming that the ghost-mode subtraction does not introduce any additional numerical error.

\begin{figure}[H]
	\centering
	\begin{subfigure}[b]{0.48\textwidth}
		\centering
		\includegraphics[width=\linewidth]{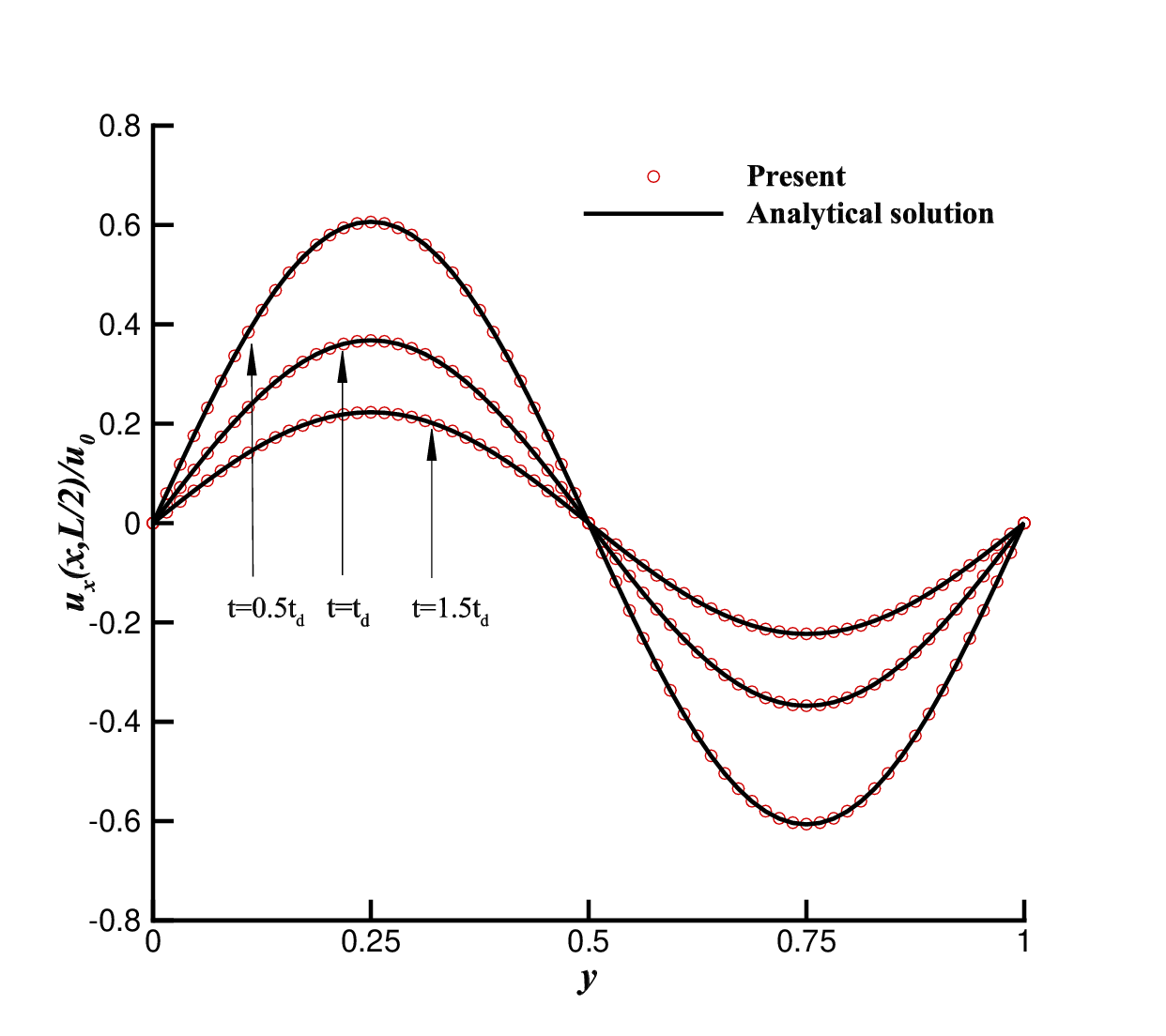}
	\end{subfigure}
	\hfill
	\begin{subfigure}[b]{0.48\textwidth}
		\centering
		\includegraphics[width=\linewidth]{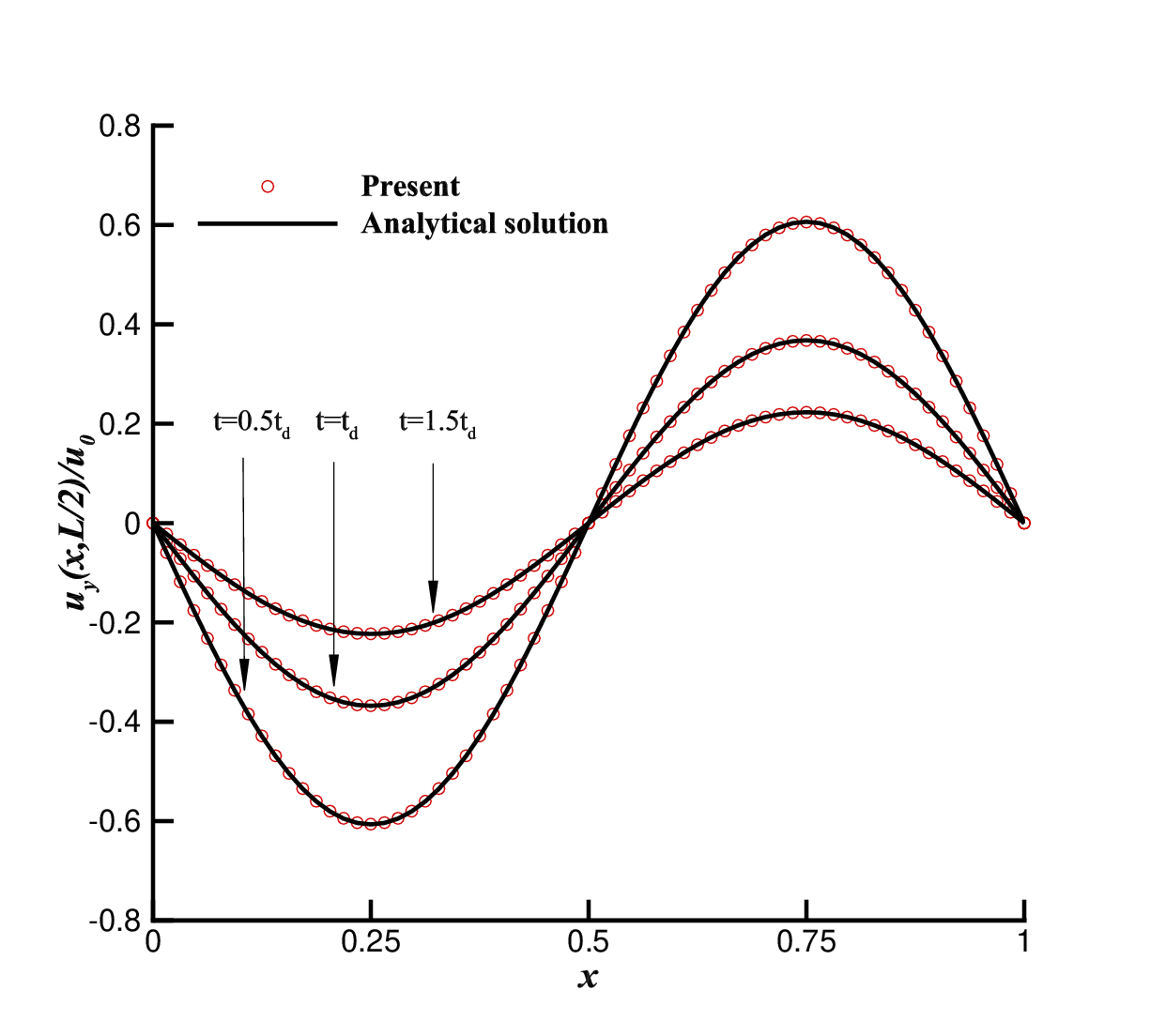}
	\end{subfigure}
	\caption{Comparison between the velocity components $u_x$ (left) and $u_y$ (right) computed by the P-TRT model and the analytical solution for the 2D decaying Taylor--Green vortex flow on a $64 \times 64$ grid with $\nu = 0.01$ at different times.}
	\label{fig:tg_solution}
\end{figure}

A grid convergence study was also performed with five resolutions: $N^2 = 32^2, 64^2, 128^2, 256^2$, and $512^2$. The computed $L_2$ errors are presented in Figure~\ref{fig:grid_convergence}. The P-TRT model maintains the second-order spatial convergence of the original TRT-RLB formulation, and the error magnitudes are consistent across the entire range of grid resolutions. This result is expected from the equivalence theorems: since the two models produce the identical post-collision distribution, they must exhibit the same truncation error structure.

\begin{figure}[H]
	\centering
	\includegraphics[width=0.5\linewidth]{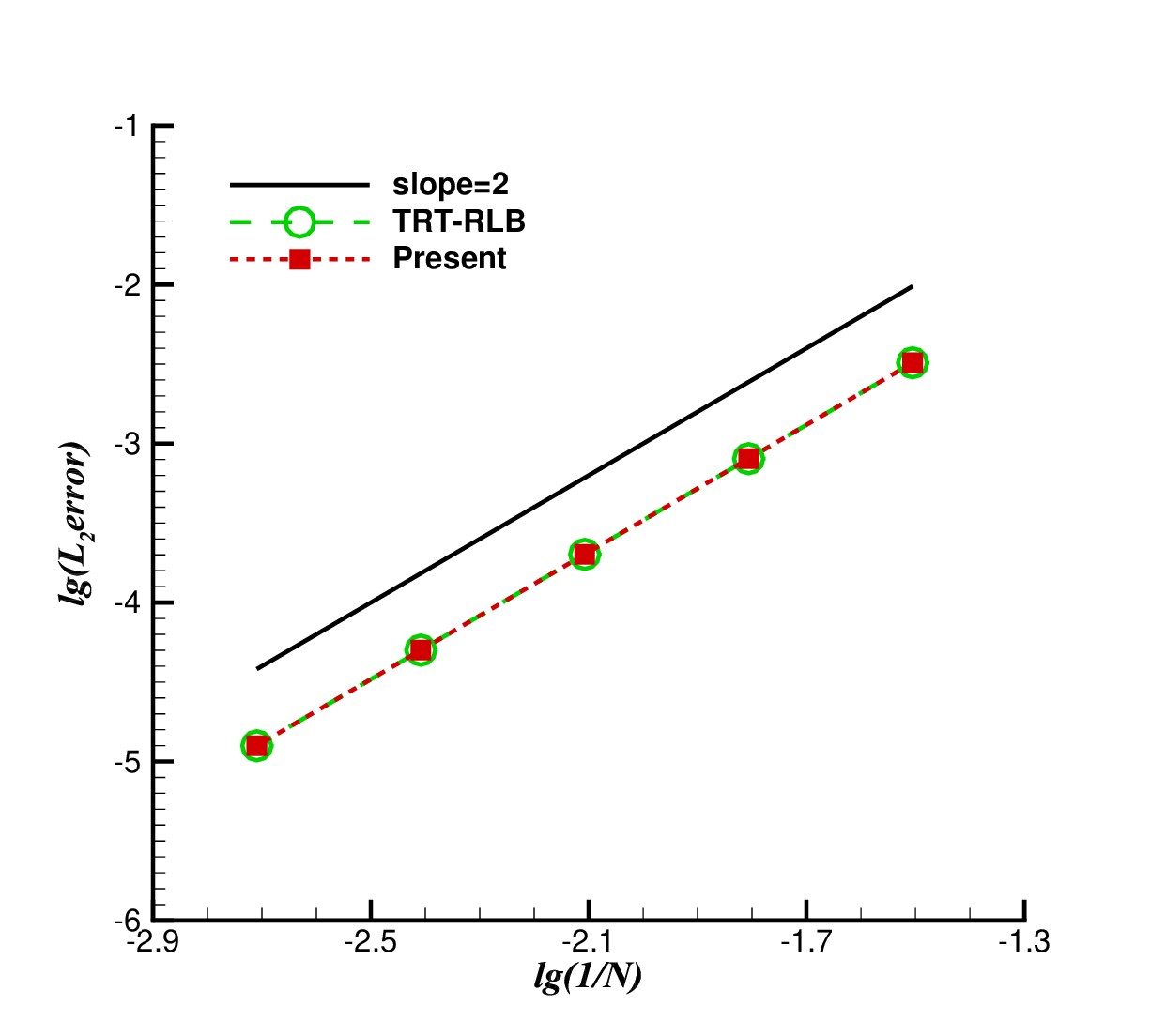}
	\caption{Grid convergence study showing the $L_2$ error norms of the TRT-RLB and P-TRT models for the 2D decaying Taylor--Green vortex flow ($\nu = 0.01$) at $t = t_d$.}
	\label{fig:grid_convergence}
\end{figure}

\subsection{Force-Driven Poiseuille Flow: Elimination of Numerical Slip}

A key feature of the TRT framework is its ability to eliminate the viscosity-dependent numerical slip inherent in the standard half-way bounce-back (HWBB) scheme \cite{ladd1994numerical}. It has been confirmed that the original TRT-RLB model effectively eliminates this numerical slip \cite{WOS:001361481500001}. Since the P-TRT model is mathematically equivalent to TRT-RLB, it must inherit this property. To verify this numerically, the 2D force-driven Poiseuille flow between two parallel plates is simulated.

We consider a steady, fully developed laminar flow driven by a constant body force. Under these conditions \cite{he1997analytic}:
\begin{equation}
	\rho = \text{const.}, \quad \partial_x u_x = 0, \quad \partial_x u_y = 0,
\end{equation}
with $u_y = 0$ throughout the domain. The body force acts solely in the streamwise direction: $F_x = \rho g$, $F_y = 0$. Under the HWBB scheme, the discrete analytical solution for the velocity profile is:
\begin{equation}
	u = \frac{4u_c}{N^2} y_i(N - y_i),
\end{equation}
where $N = L/\Delta x$, and $y_i = i - 0.5$ with $i = 1, \dots, N$ reflects the half-grid offset of the HWBB scheme. The centerline velocity is $u_c = L^2 g/(8\nu)$.

In the numerical experiments, $L = 1$, the grid size is $N_x \times N_y = 4 \times 32$, $c = 1$, $\rho_0 = 1$, $u_c = 0.1$, and simulations are performed at $Re = 1$. First, the influence of the magic parameter $\Lambda_s$ on numerical accuracy is examined. Figure~\ref{fig:Eliminate_Numerical} shows the $L_2$ error as a function of $\Lambda_s$. The numerical slip is completely eliminated when $\Lambda_s = 3/16$, consistent with the theoretical prediction for the TRT-RLB framework \cite{WOS:001361481500001}. This confirms that the ghost-mode purification preserves the slip-elimination mechanism of the TRT family: the specific relationship between $\tau_{s,1}$ and $\tau_{s,2}$ that cancels the viscosity-dependent boundary error is unaffected by the regularization procedure.

\begin{figure}[htbp]
	\centering
	\includegraphics[width=0.5\linewidth]{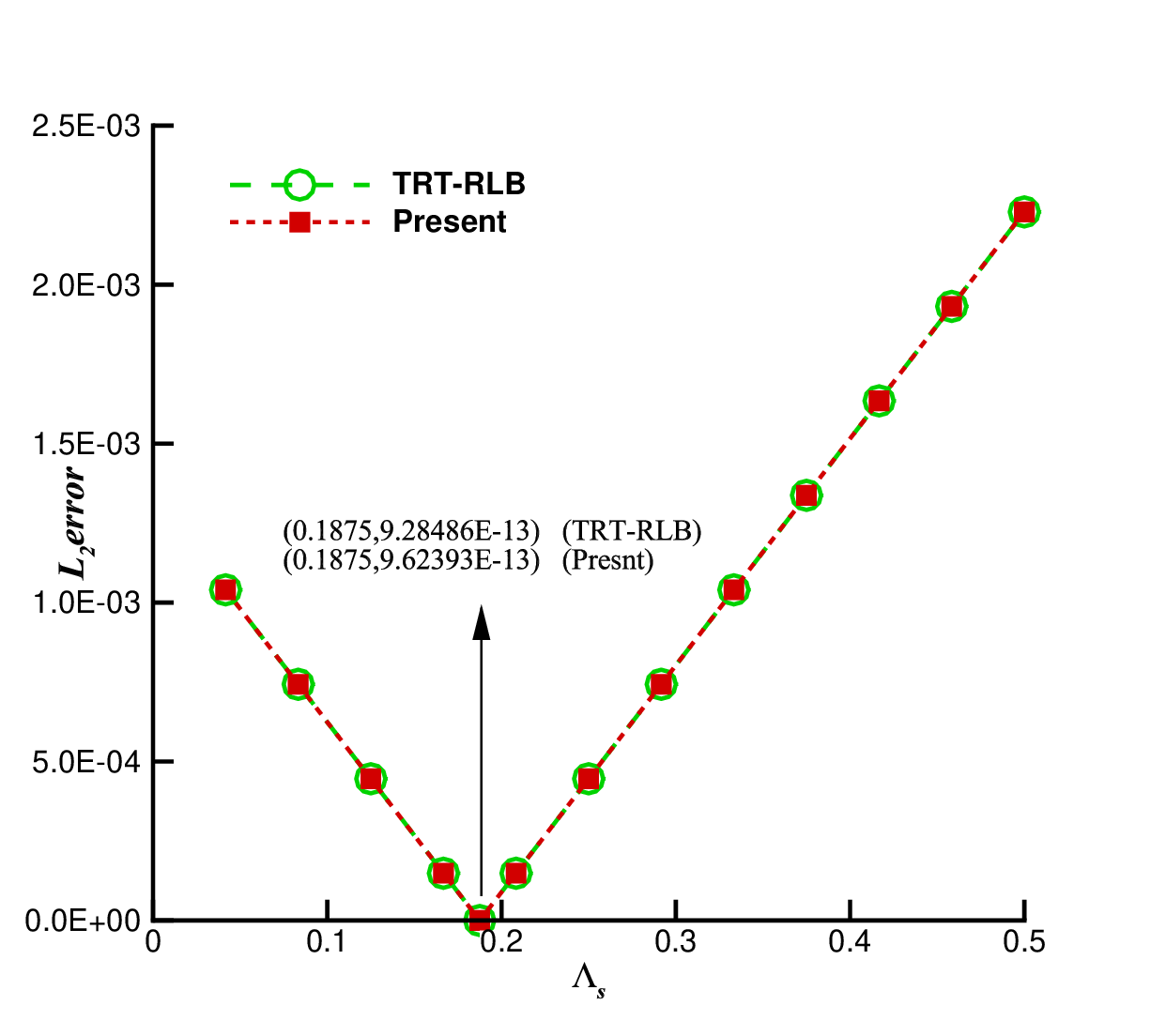}
	\caption{The $L_2$ error versus the magic parameter $\Lambda_s$ for the P-TRT model. The vanishing error at $\Lambda_s = 3/16$ confirms the slip-elimination property inherited from the TRT framework.}
	\label{fig:Eliminate_Numerical}
\end{figure}

Furthermore, simulations are performed at $Re = 1, 5$, and $10$ with $\Lambda_s = 3/16$. Figure~\ref{fig:Re_1_10} compares the velocity profiles with the analytical solutions. Table~\ref{table:p_l2} lists the $L_2$ errors, which are consistent with the original TRT-RLB model to machine precision. This is a direct consequence of the mathematical equivalence: the two models produce the same post-collision distribution and therefore the same steady-state solution.

\begin{figure}[htbp]
	\centering
	\includegraphics[width=0.5\linewidth]{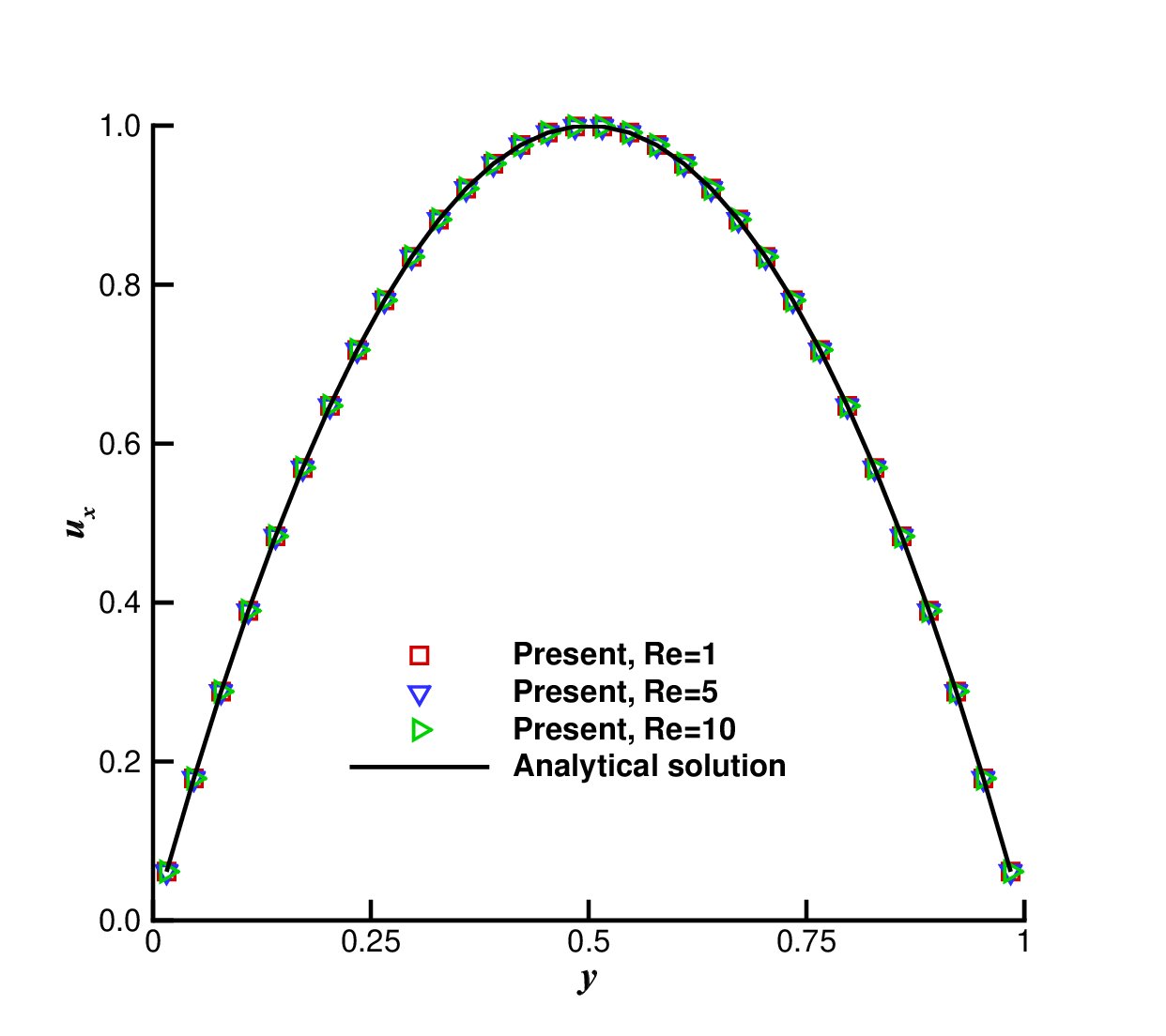}
	\caption{Comparison of horizontal velocity profiles between the P-TRT model and the analytical solution at $Re = 1, 5, 10$.}
	\label{fig:Re_1_10}
\end{figure}

\begin{table}[htbp]
	\centering
	\caption{Comparison of $L_2$ errors between the original TRT-RLB model and the P-TRT model for force-driven Poiseuille flow.}
	\small
	\begin{tabular}{c c c c}
		\toprule
		\multirow{2}{*}{$\mathrm{Re}$} & \multirow{2}{*}{$\tau_{s,1}$} & \multicolumn{2}{c}{$L_2$ error} \\
		\cmidrule(lr){3-4}
		& & TRT-RLB & P-TRT \\
		\midrule
		1   & 10.1 & $9.3\times 10^{-13}$ & $9.6\times 10^{-13}$ \\
		5   & 2.42 & $5.3\times 10^{-13}$ & $4.2\times 10^{-13}$ \\
		10  & 1.46 & $1.2\times 10^{-12}$ & $1.1\times 10^{-12}$ \\
		\bottomrule
	\end{tabular}
	\label{table:p_l2}
\end{table}

\subsection{Creeping Flow Past a Square Cylinder: Robustness}

To verify the robustness of the P-TRT model at low Reynolds numbers, we simulate the creeping flow past a square cylinder. The computational domain, illustrated in Figure~\ref{fig:square_flow}, is a square region with side length $L = 50D$, where $D$ is the side length of the cylinder. The cylinder is centered in the domain. A uniform velocity $u_c$ is imposed at the left boundary using the HWBB scheme, a constant pressure outlet is implemented at the right boundary via the Zou--He method, and periodic conditions are applied to the top and bottom boundaries. The no-slip condition on the cylinder surface is enforced by the HWBB scheme. The simulation is governed by three dimensionless parameters: $Re$, $Ma$, and the incompressibility parameter introduced by Gsell et al.\ \cite{gsell2021lattice}:
\begin{equation}
	\mathcal{T} = \frac{Ma^2}{Re}.
\end{equation}

\begin{figure}[h]
	\centering
	\includegraphics[width=0.7\linewidth]{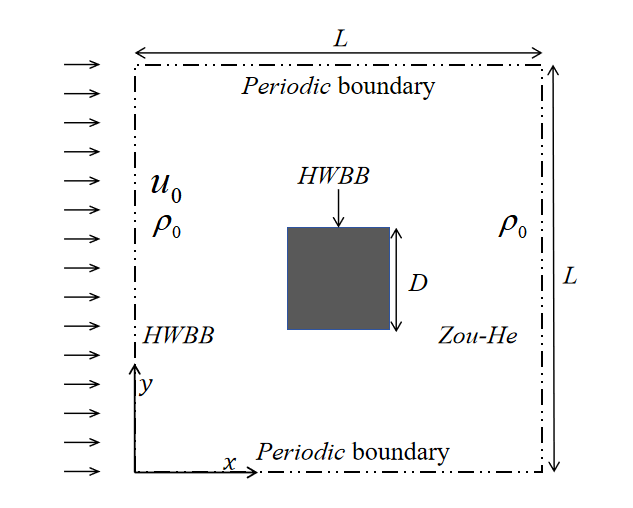}
	\caption{Schematic of the computational domain and boundary conditions for the creeping flow past a square cylinder.}
	\label{fig:square_flow}
\end{figure}

In this section, the convergence criterion is set to $L_2 < 10^{-9}$ to accommodate the small characteristic velocities. The parameters are $\rho_0 = 1$, $c = 1$, $\Delta x = 1$, and $\Lambda_s = 3/16$. The investigation consists of three parts: grid independence analysis, accuracy validation, and performance at ultra-low Reynolds numbers.

First, a grid independence test was conducted at $Re = 1$ and $Ma = 0.1$ with six grid resolutions. The drag coefficient is defined as:
\begin{equation}
	C_d = \frac{2F_x}{\rho_0 u_0^2 D},
\end{equation}
where $F_x$ is computed via the momentum exchange method:
\begin{equation}
	F_x = \frac{\Delta x^2}{\Delta t} \sum_{i \in I_w} e_{ix} \left[ f_{\bar{i}}^*(\boldsymbol{x}_s, t) + f_i^*(\boldsymbol{x}_f, t) \right].
\end{equation}
Here, $I_w$ denotes the set of discrete directions pointing from the fluid nodes toward the wall, $\boldsymbol{x}_f$ is a fluid node adjacent to the wall, and $\boldsymbol{x}_s$ the corresponding solid node. The results are shown in Figure~\ref{subfig:grid}. The drag coefficient converges with increasing resolution, and the P-TRT results are identical to those of TRT-RLB. A resolution of $D = 21\Delta x$ is adopted for subsequent simulations.

Subsequently, we validated accuracy by simulating five cases with Reynolds numbers from 2 to 20 and comparing with the empirical drag law of Sen et al.\ \cite{sen2011flow}:
\begin{equation}
	C_d = 0.7496 + 10.5767\, Re^{-0.66}, \quad Re \in [2, 40].
\end{equation}
The results in Figure~\ref{subfig:Re2-20} show excellent agreement with the reference data and consistency with the TRT-RLB model.

Finally, we investigated the low-Reynolds-number regime ($Re \le 0.1$) with $\mathcal{T} = 0.1$. The viscous drag coefficient $C_{d,\mu} = C_d \cdot Re$, which is independent of the Reynolds number in the Stokes limit, is recorded in Table~\ref{tab_3}. As the Reynolds number decreases from $10^{-1}$ to $10^{-6}$, $C_{d,\mu}$ converges to approximately 8.252. The P-TRT results are identical to the original TRT-RLB model at every Reynolds number tested. Furthermore, the relaxation time $\tau_{s,1}$ reaches values as high as 11503 without divergence, demonstrating robustness far exceeding single-relaxation-time models such as BGK and RLB, which typically require $\tau_{s,1} < 3$ for stability \cite{gsell2021lattice}. This robustness originates from the explicit removal of ghost modes in the P-TRT collision operator: even at extreme values of $\tau_{s,1}$, the purified collision step prevents the accumulation of non-physical modes that would otherwise destabilize the computation.

\begin{figure}[H]
	\centering
	\begin{subfigure}[b]{0.48\textwidth}
		\centering
		\includegraphics[width=\linewidth]{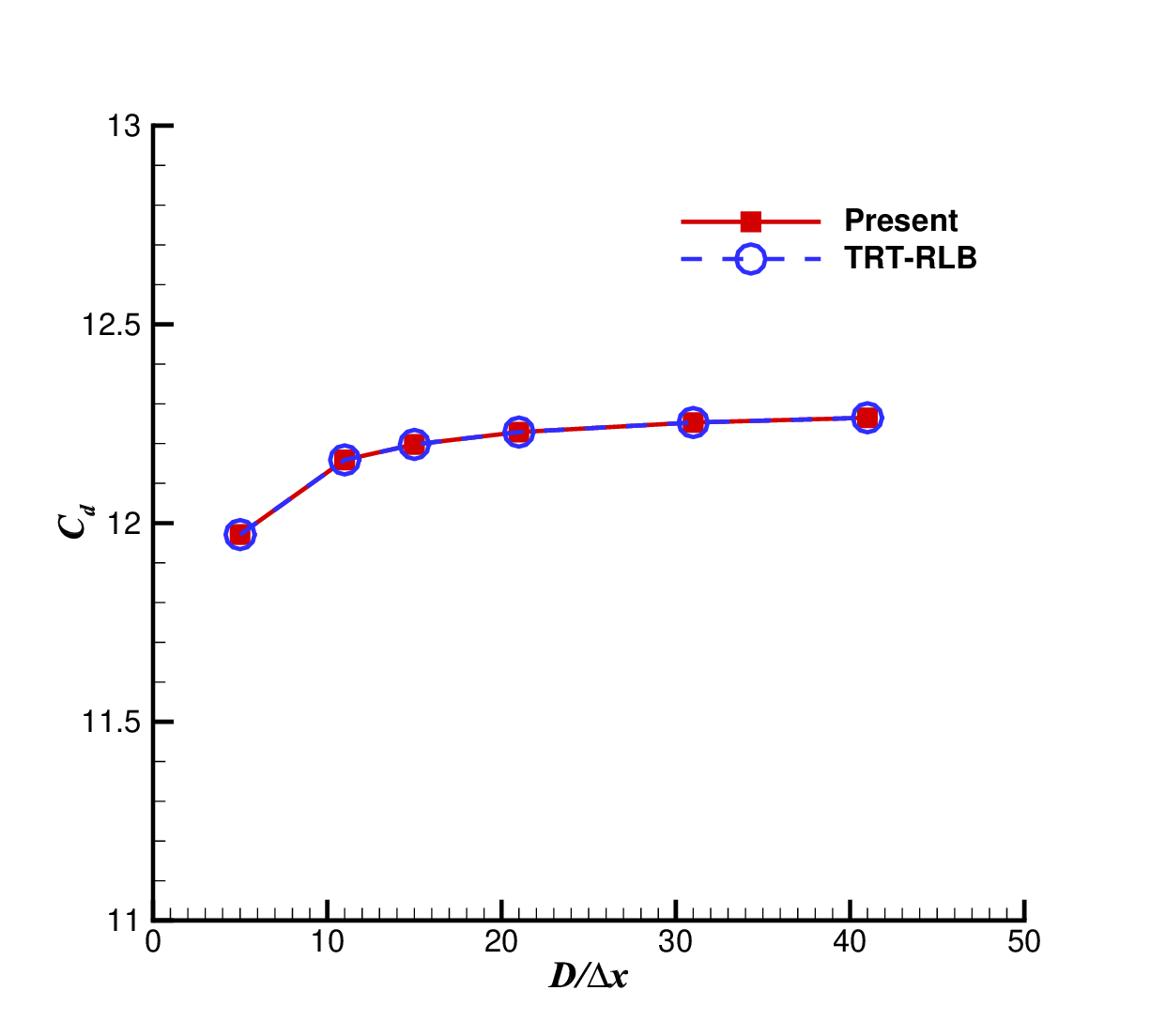}
		\caption{}
		\label{subfig:grid}
	\end{subfigure}
	\hfill
	\begin{subfigure}[b]{0.48\textwidth}
		\centering
		\includegraphics[width=\linewidth]{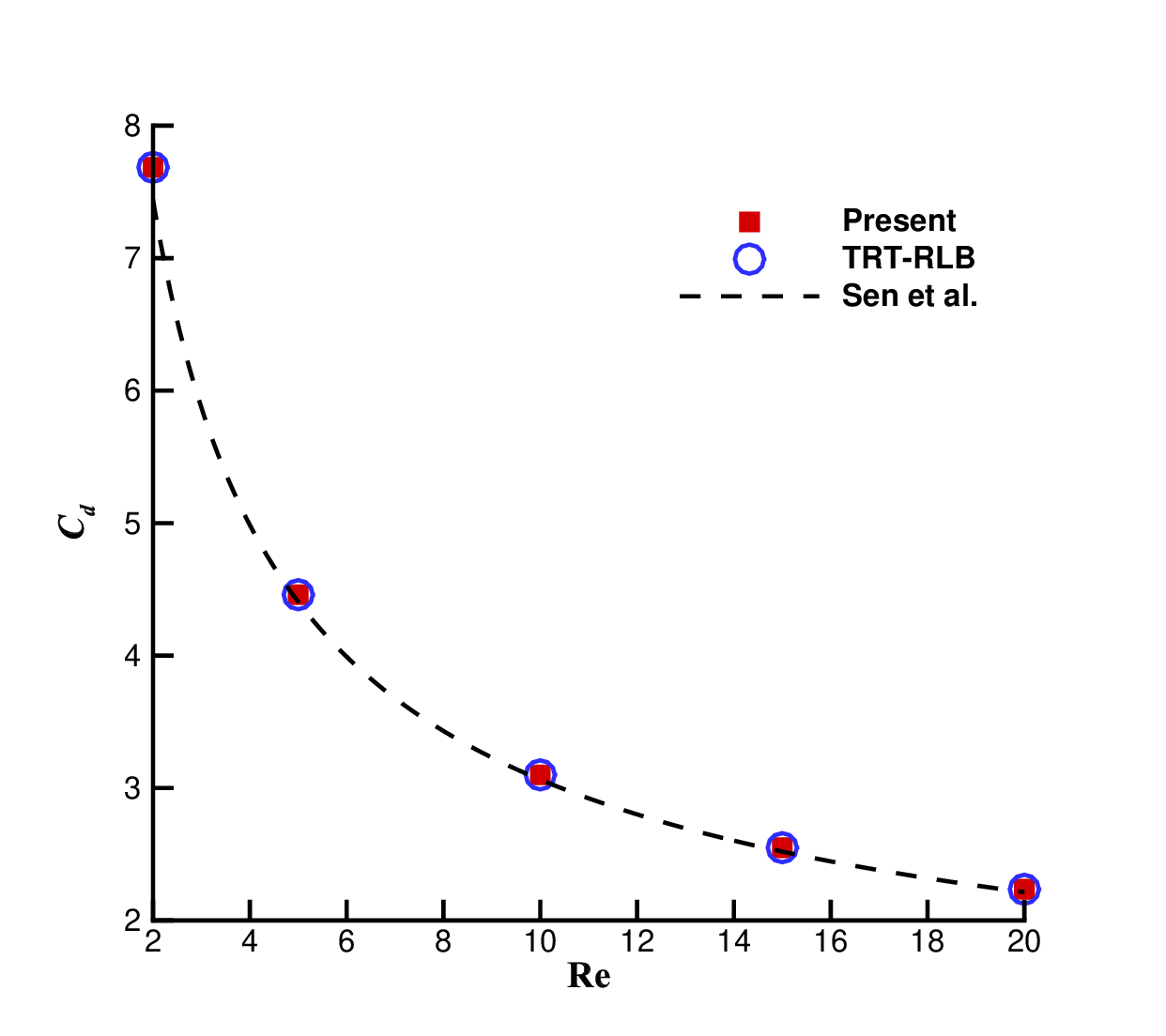}
		\caption{}
		\label{subfig:Re2-20}
	\end{subfigure}
	\caption{Numerical analysis of the flow past a square cylinder using the P-TRT model. (a) Convergence of the drag coefficient $C_d$ with increasing grid resolution $D/\Delta x$ ($Re=1, Ma=0.1$); (b) comparison of the computed $C_d$ with reference data at different Reynolds numbers ($D/\Delta x = 21$). Both sets of results are identical to those of the original TRT-RLB model.}
	\label{fig:Square_Cylinder}
\end{figure}

\begin{table}[htbp]
	\centering
	\begin{tabular}{l c c c c}
		\hline
		\multirow{2}{*}{$\mathrm{Re}$} & \multirow{2}{*}{$\mathrm{Ma}$} & \multirow{2}{*}{$\tau_{s,1}$} &
		\multicolumn{2}{c}{$C_{d,\mu}$} \\
		\cline{4-5}
		& & & TRT-RLB & P-TRT \\
		\hline
		$10^{-1}$ & $1.00 \times 10^{-1}$ & 37 & 8.449 & 8.449 \\
		$10^{-2}$ & $3.16 \times 10^{-2}$ & 116 & 8.260 & 8.260 \\
		$10^{-3}$ & $1.00 \times 10^{-2}$ & 364 & 8.253 & 8.253 \\
		$10^{-4}$ & $3.16 \times 10^{-3}$ & 1151 & 8.253 & 8.253 \\
		$10^{-5}$ & $1.00 \times 10^{-3}$ & 3638 & 8.252 & 8.252 \\
		$10^{-6}$ & $3.16 \times 10^{-4}$ & 11503 & 8.252 & 8.252 \\
		\hline
	\end{tabular}
	\caption{Comparison of the viscous drag coefficient $C_{d,\mu}$ between TRT-RLB and P-TRT at low Reynolds numbers. Parameters: $\mathcal{T} = 0.1$, $D/\Delta x = 21$, $L = 50D$.}
	\label{tab_3}
\end{table}

\section{Conclusions}
\label{sec:4}

In this work, we have established a rigorous theoretical connection between the regularized lattice Boltzmann method and the ghost-mode structure of the two-relaxation-time collision operator. By invoking the orthogonal decomposition of the discrete velocity space into hydrodynamic and non-hydrodynamic subspaces, we proved that the Hermite-based regularization employed in the TRT-RLB model of Yu et al.\ \cite{WOS:001361481500001} is mathematically equivalent to the explicit removal of ghost modes from the standard TRT collision operator. Specifically, the second-order Hermite projection reduces to the symmetric non-equilibrium part minus the ghost-mode contribution (Theorem~\ref{thm:2nd}), and the third-order Hermite projection reduces to the antisymmetric non-equilibrium part under momentum conservation (Theorem~\ref{thm:3rd}). For the D2Q9 and D3Q19 lattices, both equivalences are exact, meaning that the two formulations produce the identical post-collision distribution at every grid point and time step.

Based on this equivalence, we proposed the Purified TRT (P-TRT) model, which upgrades the standard TRT by explicitly filtering out its ghost modes through simple algebraic operations. The P-TRT model offers a dual interpretation depending on the user's perspective:

For TRT practitioners, it provides a minimal and physically transparent modification---subtracting a small number of identifiable ghost-mode contributions from the symmetric collision term---that elevates the stability of the classical TRT scheme to the regularization level without sacrificing its algorithmic simplicity or boundary-accuracy properties.

For TRT-RLB users, it provides a mathematically equivalent but substantially more efficient algorithm that replaces the expensive Hermite tensor projections with parity decomposition and scalar ghost-mode subtraction, achieving a 71\% reduction in non-equilibrium collision operations for the D2Q9 and D3Q19 lattices.

Extensive numerical validation through four benchmark problems confirmed that the P-TRT model preserves all the desirable properties of the original TRT-RLB formulation: robust stability at Reynolds numbers up to $10^7$, second-order spatial convergence, exact elimination of viscosity-dependent numerical slip at bounce-back boundaries ($\Lambda_s = 3/16$), and reliable performance at ultra-low Reynolds numbers with relaxation times exceeding $10^4$. The measured wall-clock speedup of approximately 30\% for the D2Q9 lattice is consistent with the theoretical FLOP reduction when the collision step dominates the computational cost.

Beyond the immediate practical benefits, a linear stability analysis in the moment space reveals that the P-TRT collision operator annihilates the ghost eigenvalue ($z_G = 0$), thereby completely decoupling numerical stability from the non-hydrodynamic sector (Theorem~\ref{thm:stability}). This provides a rigorous theoretical explanation for the superior stability observed in the numerical benchmarks: the spectral radius of the P-TRT evolution operator is bounded above by that of the standard TRT (Remark~\ref{rmk:spectral}), and the stability bottleneck is shifted from the ghost mode to the physically meaningful hydrodynamic modes. The analytical determination of the optimal free parameter $\tau_{s,2}$, which requires a parametric eigenvalue study of the hydrodynamic evolution matrix over the full Brillouin zone, is identified as an important direction for future work.

More broadly, the present analysis reveals a deeper insight: the core mechanism by which regularization improves the stability of LBM collision operators is the annihilation of non-hydrodynamic ghost modes. This understanding, obscured in the original Hermite tensor formalism, becomes transparent in the orthogonal decomposition framework. We believe that this ghost-mode perspective may prove useful for the analysis and design of other advanced collision operators in the lattice Boltzmann community.

\section*{Acknowledgments}
This work was supported by the National Natural Science Foundation of China (Grant Nos. 12101527, 12271464, and 12371373), the Natural Science Foundation of Hunan Province (Grant No. 2026JJ60005), the 111 Project (No. D23017), and the Program for Science and Technology Innovative Research Team in Higher Educational Institutions of Hunan Province. Computational resources were provided by the High Performance Computing Platform of Xiangtan University.

\appendix

\section{Ghost-Mode Enumeration and Explicit Formulas}
\label{app:ghost_modes}

This appendix provides the complete enumeration of ghost modes and the explicit implementation formulas for the P-TRT model on the D2Q9, D3Q19, and D3Q27 lattices. Table~\ref{tab:lattice_summary} summarizes the key structural parameters.

\begin{table}[h]
	\centering
	\caption{Structural parameters of common lattice models. Here $Q$ is the number of discrete velocities, $Q_0$ the number of rest velocities, $Q_p$ the number of inversion pairs, $N_G$ the number of symmetric ghost modes, and $N_3$ the number of independent third-order Hermite polynomials.}
	\label{tab:lattice_summary}
	\begin{tabular}{l c c c c c c c c}
		\toprule
		Lattice & $Q$ & $D$ & $c_s^2$ & $Q_0$ & $Q_p$ & $\dim(\mathcal{S})$ & $\dim(\mathcal{A})$ & $N_G$ / $N_3$ \\
		\midrule
		D2Q9  & 9  & 2 & $c^2/3$ & 1 & 4  & 5  & 4  & 1 / 2 \\
		D3Q19 & 19 & 3 & $c^2/3$ & 1 & 9  & 10 & 9  & 3 / 6 \\
		D3Q27 & 27 & 3 & $c^2/3$ & 1 & 13 & 14 & 13 & 7 / 7 \\
		\bottomrule
	\end{tabular}
\end{table}

\subsection{D2Q9 Lattice}
\label{app:D2Q9}

\paragraph{Symmetric ghost mode ($N_G = 1$).}
The unique ghost mode is the fourth-order product:
\begin{equation}
	\phi_i^{(G)} = (e_{ix}^2 - c_s^2)(e_{iy}^2 - c_s^2),
\end{equation}
with values at each velocity direction:
\begin{center}
	\begin{tabular}{c|ccccccccc}
		$i$ & 0 & 1 & 2 & 3 & 4 & 5 & 6 & 7 & 8 \\
		\hline
		$\phi_i^{(G)}$ & $c_s^4$ & $-2c_s^4$ & $-2c_s^4$ & $-2c_s^4$ & $-2c_s^4$ & $4c_s^4$ & $4c_s^4$ & $4c_s^4$ & $4c_s^4$
	\end{tabular}
\end{center}
The weighted norm is $N_G^{(\phi)} = \sum_i w_i [\phi_i^{(G)}]^2 = 4c_s^8$. The ghost scalar moment and projection are:
\begin{equation}
	S_G = \sum_{j=0}^{8} \phi_j^{(G)} f_j^{\mathrm{neq}} = c_s^4 \left[ f_0^{\mathrm{neq}} - 2 \sum_{j=1}^{4} f_j^{\mathrm{neq}} + 4 \sum_{j=5}^{8} f_j^{\mathrm{neq}} \right],
\end{equation}
\begin{equation}
	(\hat{P}^{(G)} f^{\mathrm{neq}})_i = \frac{w_i \, \phi_i^{(G)}}{4 c_s^8} \, S_G.
\end{equation}

\paragraph{Antisymmetric subspace ($\dim(\mathcal{A}) = 4 = D + N_3$).}
No antisymmetric ghost modes exist. The two independent third-order Hermite modes are:
\begin{equation}
	\psi_i^{(xxy)} = e_{iy}(e_{ix}^2 - c_s^2), \qquad \psi_i^{(xyy)} = e_{ix}(e_{iy}^2 - c_s^2),
\end{equation}
with norms $N^{(xxy)} = N^{(xyy)} = 2c_s^6$. The purely cubic modes $\mathcal{H}^{(3)}_{xxx}$ and $\mathcal{H}^{(3)}_{yyy}$ vanish due to the lattice aliasing $e_{i\alpha}^3 = c^2 e_{i\alpha}$.

\paragraph{P-TRT implementation.}
The second-order regularization term is:
\begin{equation}
	\mathcal{R}_i^{(2)} = f_i^{\mathrm{neq},+} - \frac{w_i \, \phi_i^{(G)}}{4c_s^8} \, S_G.
\end{equation}
The third-order regularization term (without external force) is:
\begin{equation}
	\mathcal{R}_i^{(3)} = f_i^{\mathrm{neq},-}.
\end{equation}
Both relations are exact.

\subsection{D3Q19 Lattice}
\label{app:D3Q19}

The D3Q19 lattice consists of 1 rest velocity, 6 axial velocities (3 inversion pairs), and 12 face-diagonal velocities (6 inversion pairs). No body-diagonal velocities are present.

\paragraph{Symmetric ghost modes ($N_G = 3$).}
Three orthogonal fourth-order ghost modes exist:
\begin{align}
	\phi_i^{(G_1)} &= (e_{ix}^2 - c_s^2)(e_{iy}^2 - c_s^2), \\
	\phi_i^{(G_2)} &= (e_{iy}^2 - c_s^2)(e_{iz}^2 - c_s^2), \\
	\phi_i^{(G_3)} &= (e_{iz}^2 - c_s^2)(e_{ix}^2 - c_s^2).
\end{align}
These modes are mutually orthogonal under the weighted inner product. Each has the weighted norm $N_{G,k}^{(\phi)} = 4c_s^8$. The absence of body-diagonal velocities ($e_{ix}, e_{iy}, e_{iz}$ simultaneously nonzero) ensures the orthogonality and the fact that no higher-order symmetric ghost modes arise.

\paragraph{Antisymmetric subspace ($\dim(\mathcal{A}) = 9 = D + N_3$).}
No antisymmetric ghost modes exist. The six independent third-order Hermite modes are:
\begin{align}
	&\psi_i^{(xxy)} = e_{iy}(e_{ix}^2 - c_s^2), \quad \psi_i^{(xxz)} = e_{iz}(e_{ix}^2 - c_s^2), \quad \psi_i^{(xyy)} = e_{ix}(e_{iy}^2 - c_s^2), \\
	&\psi_i^{(yyz)} = e_{iz}(e_{iy}^2 - c_s^2), \quad \psi_i^{(xzz)} = e_{ix}(e_{iz}^2 - c_s^2), \quad \psi_i^{(yzz)} = e_{iy}(e_{iz}^2 - c_s^2).
\end{align}
The mixed mode $\mathcal{H}^{(3)}_{xyz}$ vanishes on D3Q19 because no velocity has all three components nonzero.

\paragraph{P-TRT implementation.}
The formulas are structurally identical to the D2Q9 case:
\begin{equation}
	\mathcal{R}_i^{(2)} = f_i^{\mathrm{neq},+} - \sum_{k=1}^{3} \frac{w_i \, \phi_{i,k}^{(G)}}{4c_s^8} \sum_{j} \phi_{j,k}^{(G)} f_j^{\mathrm{neq}}, \qquad \mathcal{R}_i^{(3)} = f_i^{\mathrm{neq},-}.
\end{equation}
Both relations are exact.

\subsection{D3Q27 Lattice}
\label{app:D3Q27}

The D3Q27 lattice consists of 1 rest velocity, 6 axial velocities, 12 face-diagonal velocities, and 8 body-diagonal velocities ($Q_0 = 1$, $Q_p = 13$).

\paragraph{Symmetric ghost modes ($N_G = 7$).}
Seven orthogonal ghost modes exist, organized into two groups:

\emph{Group~1: Fourth-order diagonal coupling modes (3) and the sixth-order full coupling mode (1).}
\begin{align}
	\phi_i^{(G_1)} &= (e_{ix}^2 - c_s^2)(e_{iy}^2 - c_s^2), \\
	\phi_i^{(G_2)} &= (e_{iy}^2 - c_s^2)(e_{iz}^2 - c_s^2), \\
	\phi_i^{(G_3)} &= (e_{iz}^2 - c_s^2)(e_{ix}^2 - c_s^2), \\
	\phi_i^{(G_4)} &= (e_{ix}^2 - c_s^2)(e_{iy}^2 - c_s^2)(e_{iz}^2 - c_s^2) / c_s^2.
\end{align}

\emph{Group~2: Fourth-order shear--energy coupling modes (3).}
\begin{align}
	\phi_i^{(G_5)} &= e_{ix} e_{iy} (e_{iz}^2 - c_s^2), \\
	\phi_i^{(G_6)} &= e_{iy} e_{iz} (e_{ix}^2 - c_s^2), \\
	\phi_i^{(G_7)} &= e_{iz} e_{ix} (e_{iy}^2 - c_s^2).
\end{align}

\paragraph{Antisymmetric subspace ($\dim(\mathcal{A}) = 13 > D + N_3 = 10$).}
The deficit of 3 dimensions corresponds to fifth-order antisymmetric modes:
\begin{equation}
	\psi_i^{(5,x)} = e_{ix}(e_{iy}^2 - c_s^2)(e_{iz}^2 - c_s^2),
\end{equation}
and its cyclic permutations $\psi_i^{(5,y)}$ and $\psi_i^{(5,z)}$. These modes are nonzero only at body-diagonal velocities.

\paragraph{P-TRT implementation (hybrid strategy)}
Since $N_G = 7 > M_2 = 6$, the ghost subtraction method for the second-order term is more expensive than the direct Hermite projection. Therefore, a hybrid strategy is recommended for D3Q27:
\begin{equation}
	\mathcal{R}_i^{(2)} = w_i \frac{\mathcal{H}_{i,\alpha\beta}^{(2)}}{2c_s^4} \mathcal{A}_{\alpha\beta}^{\mathrm{neq}} \quad \text{(tensor projection, exact)},
\end{equation}
\begin{equation}
	\mathcal{R}_i^{(3)} \approx f_i^{\mathrm{neq},-} \quad \text{(antisymmetric approximation)}.
\end{equation}
The approximation in $\mathcal{R}_i^{(3)}$ neglects the fifth-order residual, which is $O(\mathrm{Kn}^3)$ in the Chapman--Enskog expansion and negligibly small in practice. If exact equivalence is required, one may subtract the fifth-order projection explicitly:
\begin{equation}
	\mathcal{R}_i^{(3)} = f_i^{\mathrm{neq},-} - \sum_{k \in \{x,y,z\}} (\hat{P}^{(5,k)} f^{\mathrm{neq}})_i + \frac{\Delta t}{2} \frac{w_i \, e_{i\alpha}}{c_s^2} F_\alpha.
\end{equation}


\begin{thebibliography}{10}
	\expandafter\ifx\csname url\endcsname\relax
	\def\url#1{\texttt{#1}}\fi
	\expandafter\ifx\csname urlprefix\endcsname\relax\def\urlprefix{URL }\fi
	\expandafter\ifx\csname href\endcsname\relax
	\def\href#1#2{#2} \def\path#1{#1}\fi
	
	\bibitem{oh2025numerical}
	H.~Oh, H.~Jo, Numerical study of turbulent flow boiling heat transfer in
	structured cooling channels using lattice Boltzmann method with advanced
	outlet boundary conditions, Physics of Fluids 37~(1) (2025).
	
	\bibitem{montessori2024high}
	A.~Montessori, M.~La~Rocca, G.~Amati, M.~Lauricella, A.~Tiribocchi, S.~Succi,
	High-order thread-safe lattice Boltzmann model for high performance computing
	turbulent flow simulations, Physics of Fluids 36~(3) (2024).
	
	\bibitem{guo2020improved}
	S.~Guo, Y.~Feng, P.~Sagaut, Improved standard thermal lattice Boltzmann model
	with hybrid recursive regularization for compressible laminar and turbulent
	flows, Physics of Fluids 32~(12) (2020).
	
	\bibitem{suga2015D3Q27}
	K.~Suga, Y.~Kuwata, K.~Takashima, R.~Chikasue, A D3Q27 multiple-relaxation-time
	lattice Boltzmann method for turbulent flows, Computers \& Mathematics with
	Applications 69~(6) (2015) 518--529.
	
	\bibitem{lulli2024metastable}
	M.~Lulli, L.~Biferale, G.~Falcucci, M.~Sbragaglia, D.~Yang, X.~Shan, Metastable
	and unstable hydrodynamics in multiphase lattice Boltzmann, Physical Review E
	109~(4) (2024) 045304.
	
	\bibitem{huang2024three}
	R.~Huang, Q.~Li, Y.~Qiu, Three-dimensional lattice Boltzmann model with
	self-tuning equation of state for multiphase flows, Physical Review E 109~(6)
	(2024) 065306.
	
	\bibitem{zhang2025infiltration}
	X.~Zhang, T.~Huang, Z.~Ge, T.~Man, H.~E. Huppert, Infiltration characteristics
	of slurries in porous media based on the coupled lattice-Boltzmann discrete
	element method, Computers and Geotechnics 177 (2025) 106865.
	
	\bibitem{sourya2024comparative}
	D.~P. Sourya, P.~S. Gurugubelli, S.~Bhaskaran, N.~Vorhauer-Huget, E.~Tsotsas,
	V.~K. Surasani, A comparative study on the lattice Boltzmann method and the
	VOF-continuum method for oxygen transport in the anodic porous transport
	layer of an electrolyzer, International Journal of Hydrogen Energy 92 (2024)
	1091--1098.
	
	\bibitem{hosseini2024lattice}
	S.~A. Hosseini, P.~Boivin, D.~Th{\'e}venin, I.~Karlin, Lattice Boltzmann
	methods for combustion applications, Progress in Energy and Combustion
	Science 102 (2024) 101140.
	
	\bibitem{stockinger2024lattice}
	C.~Stockinger, A.~Raiolo, R.~Alamian, A.~Hadjadj, U.~Nieken, M.~S. Shadloo,
	Lattice Boltzmann simulations of heterogeneous combustion reactions for
	application in porous media, Engineering Analysis with Boundary Elements 166
	(2024) 105817.
	
	\bibitem{lei2021study}
	T.~Lei, Z.~Wang, K.~H. Luo, Study of pore-scale coke combustion in porous media
	using lattice Boltzmann method, Combustion and Flame 225 (2021) 104--119.
	
	\bibitem{tayyab2020hybrid}
	M.~Tayyab, S.~Zhao, Y.~Feng, P.~Boivin, Hybrid regularized lattice-Boltzmann
	modelling of premixed and non-premixed combustion processes, Combustion and
	Flame 211 (2020) 173--184.
	
	\bibitem{fei2023lattice}
	L.~Fei, F.~Qin, J.~Zhao, D.~Derome, J.~Carmeliet, Lattice Boltzmann modelling
	of isothermal two-component evaporation in porous media, Journal of Fluid
	Mechanics 955 (2023) A18.
	
	\bibitem{kumar2023lattice}
	S.~Kumar, D.~Panda, P.~Ghodke, K.~M. Gangawane, Lattice Boltzmann method for
	heat transfer in phase change materials: a review, Journal of Thermal
	Analysis and Calorimetry 148~(17) (2023) 9263--9287.
	
	\bibitem{li2016pinning}
	Q.~Li, P.~Zhou, H.~Yan, Pinning--depinning mechanism of the contact line during
	evaporation on chemically patterned surfaces: A lattice Boltzmann study,
	Langmuir 32~(37) (2016) 9389--9396.
	
	\bibitem{cherkaoui2023numerical}
	I.~Cherkaoui, S.~Bettaibi, A.~Barkaoui, F.~Kuznik, Numerical study of pulsatile
	thermal magnetohydrodynamic blood flow in an artery with aneurysm using
	lattice Boltzmann method (LBM), Communications in Nonlinear Science and
	Numerical Simulation 123 (2023) 107281.
	
	\bibitem{cherkaoui2022magnetohydrodynamic}
	I.~Cherkaoui, S.~Bettaibi, A.~Barkaoui, F.~Kuznik, Magnetohydrodynamic blood
	flow study in stenotic coronary artery using lattice Boltzmann method,
	Computer Methods and Programs in Biomedicine 221 (2022) 106850.
	
	\bibitem{esfahani2020lattice}
	S.~S. Esfahani, X.~Zhai, M.~Chen, A.~Amira, F.~Bensaali, J.~AbiNahed, S.~Dakua,
	G.~Younes, A.~Baobeid, R.~A. Richardson, et~al., Lattice-Boltzmann
	interactive blood flow simulation pipeline, International Journal of Computer
	Assisted Radiology and Surgery 15~(4) (2020) 629--639.
	
	\bibitem{tiribocchi2025lattice}
	A.~Tiribocchi, M.~Durve, M.~Lauricella, A.~Montessori, J.-M. Tucny, S.~Succi,
	Lattice Boltzmann simulations for soft flowing matter, Physics Reports 1105
	(2025) 1--52.
	
	\bibitem{lallemand2021lattice}
	P.~Lallemand, L.-S. Luo, M.~Krafczyk, W.-A. Yong, The lattice Boltzmann method
	for nearly incompressible flows, Journal of Computational Physics 431 (2021)
	109713.
	
	\bibitem{kruger2017lattice}
	T.~Kr{\"u}ger, H.~Kusumaatmaja, A.~Kuzmin, O.~Shardt, G.~Silva, E.~M. Viggen,
	The lattice Boltzmann method, Springer, 2017.
	
	\bibitem{guo2013lattice}
	Z.~Guo, C.~Shu, Lattice Boltzmann method and its application in engineering,
	Vol.~3, World Scientific, 2013.
	
	\bibitem{bhatnagar1954model}
	P.~L. Bhatnagar, E.~P. Gross, M.~Krook, A model for collision processes in
	gases. i. small amplitude processes in charged and neutral one-component
	systems, Physical Review 94~(3) (1954) 511.
	
	\bibitem{qian1992lattice}
	Y.-H. Qian, D.~d'Humi{\`e}res, P.~Lallemand, Lattice BGK models for
	Navier-Stokes equation, Europhysics Letters 17~(6) (1992) 479.
	
	\bibitem{lallemand2000theory}
	P.~Lallemand, L.-S. Luo, Theory of the lattice {B}oltzmann method: Dispersion,
	dissipation, isotropy, {G}alilean invariance, and stability, Physical Review
	E 61~(6) (2000) 6546.
	
	\bibitem{MRT1}
	D.~d'Humi{\`e}res, Multiple--relaxation--time lattice Boltzmann models in three
	dimensions, Philosophical Transactions of the Royal Society of London. Series
	A: Mathematical, Physical and Engineering Sciences 360~(1792) (2002)
	437--451.
	
	\bibitem{du2006multi}
	R.~Du, B.~Shi, X.~Chen, Multi-relaxation-time lattice Boltzmann model for
	incompressible flow, Physics Letters A 359~(6) (2006) 564--572.
	
	\bibitem{ning2016numerical}
	Y.~Ning, K.~N. Premnath, D.~V. Patil, Numerical study of the properties of the
	central moment lattice Boltzmann method, International Journal for Numerical
	Methods in Fluids 82~(2) (2016) 59--90.
	
	\bibitem{fei2017consistent}
	L.~Fei, K.~H. Luo, Consistent forcing scheme in the cascaded lattice Boltzmann
	method, Physical Review E 96~(5) (2017) 053--307.
	
	\bibitem{geier2006cascaded}
	M.~Geier, A.~Greiner, J.~G. Korvink, Cascaded digital lattice Boltzmann
	automata for high reynolds number flow, Physical Review E—Statistical,
	Nonlinear, and Soft Matter Physics 73~(6) (2006) 066705.
	
	\bibitem{premnath2009incorporating}
	K.~N. Premnath, S.~Banerjee, Incorporating forcing terms in cascaded lattice
	Boltzmann approach by method of central moments, Physical Review
	E--Statistical, Nonlinear, and Soft Matter Physics 80~(3) (2009) 036702.
	
	\bibitem{TRT2008study}
	I.~Ginzburg, F.~Verhaeghe, D.~d'Humieres, Study of simple hydrodynamic
	solutions with the two-relaxation-times lattice Boltzmann scheme,
	Communications in Computational Physics 3~(3) (2008) 519--581.
	
	\bibitem{TRT2008two}
	I.~Ginzburg, F.~Verhaeghe, D.~d'Humieres, Two-relaxation-time lattice Boltzmann
	scheme: About parametrization, velocity, pressure and mixed boundary
	conditions, Communications in Computational Physics 3~(2) (2008) 427--478.
	
	\bibitem{ginzburg2005equilibrium}
	I.~Ginzburg, Equilibrium-type and link-type lattice {B}oltzmann models for
	generic advection and anisotropic-dispersion equation, Advances in Water
	Resources 28~(11) (2005) 1171--1195.
	
	\bibitem{latt2006lattice}
	J.~Latt, B.~Chopard, Lattice {B}oltzmann method with regularized pre-collision
	distribution functions, Mathematics and Computers in Simulation 72~(2-6)
	(2006) 165--168.
	
	\bibitem{RLB}
	A.~Jonnalagadda, A.~Sharma, A.~Agrawal, Onsager-regularized lattice Boltzmann
	method: A nonequilibrium thermodynamics-based regularized lattice Boltzmann
	method, Physical Review E 104~(1) (2021) 015313.
	
	\bibitem{RLB1}
	C.~Coreixas, G.~Wissocq, G.~Puigt, J.-F. Boussuge, P.~Sagaut, Recursive
	regularization step for high-order lattice Boltzmann methods, Physical Review
	E 96~(3) (2017).
	
	\bibitem{RLB_yuan47}
	K.~K. Mattila, P.~C. Philippi, L.~A. Hegele, High-order regularization in
	lattice-Boltzmann equations, Physics of Fluids 29~(4) (2017).
	
	\bibitem{yang2024efficient}
	X.~Yang, Y.~Zhang, An efficient regularized lattice Boltzmann method to solve
	consistent and conservative phase field model for simulating incompressible
	two-phase flows, Advances in Applied Mathematics and Mechanics 16~(5) (2024)
	1277--1296.
	
	\bibitem{yu2026improved}
	Y.~Yu, S.~Chen, Y.~Zhou, L.~Wang, H.-Z. Yuan, S.~Shu, An improved lattice
	Boltzmann method with a novel conservative boundary scheme for viscoelastic
	fluid flows, Journal of Computational Physics (2026) 114667.
	
	\bibitem{feng2019solid}
	Y.~Feng, S.~Guo, J.~Jacob, P.~Sagaut, Solid wall and open boundary conditions
	in hybrid recursive regularized lattice Boltzmann method for compressible
	flows, Physics of Fluids 31~(12) (2019).
	
	\bibitem{WOS:001361481500001}
	Y.~Yu, Z.~Qin, S.~Chen, S.~Shu, H.~Yuan, Two-relaxation-time regularized
	lattice Boltzmann model for navier-stokes equations, Advances in Applied Mathematics and Mechanics 17~(2) (2025) 489--516.
	\newblock \href {https://doi.org/10.4208/aamm.OA-2024-0203}
	{\path{doi:10.4208/aamm.OA-2024-0203}}.
	
	\bibitem{guo2002discrete}
	Z.~Guo, C.~Zheng, B.~Shi, Discrete lattice effects on the forcing term in the
	lattice Boltzmann method, Physical Review E 65~(4) (2002) 046308.
	
	\bibitem{he1998discrete}
	X.~He, X.~Shan, G.~D. Doolen, Discrete Boltzmann equation model for nonideal
	gases, Physical Review E 57~(1) (1998) R13.
	
	\bibitem{postma2020force}
	B.~Postma, G.~Silva, Force methods for the two-relaxation-times lattice
	{{Boltzmann}}, Physical Review E 102~(6) (2020) 063307.
	\newblock \href {https://doi.org/10.1103/PhysRevE.102.063307}
	{\path{doi:10.1103/PhysRevE.102.063307}}.
	
	\bibitem{pearson1965computational}
	C.~E. Pearson, A computational method for viscous flow problems, Journal of
	Fluid Mechanics 21~(4) (1965) 611--622.
	
	\bibitem{ladd1994numerical}
	A.~J. Ladd, Numerical simulations of particulate suspensions via a discretized
	Boltzmann equation. Part 1. theoretical foundation, Journal of Fluid
	Mechanics 271 (1994) 285--309.
	
	\bibitem{he1997analytic}
	X.~He, Q.~Zou, L.-S. Luo, M.~Dembo, Analytic solutions of simple flows and
	analysis of nonslip boundary conditions for the lattice Boltzmann BGK model,
	Journal of Statistical Physics 87~(1) (1997) 115--136.
	
	\bibitem{gsell2021lattice}
	S.~Gsell, U.~d'Ortona, J.~Favier, Lattice-Boltzmann simulation of creeping
	generalized newtonian flows: theory and guidelines, Journal of Computational
	Physics 429 (2021) 109943.
	
	\bibitem{sen2011flow}
	S.~Sen, S.~Mittal, G.~Biswas, Flow past a square cylinder at low reynolds
	numbers, International Journal for Numerical Methods in Fluids 67~(9) (2011)
	1160--1174.
	
\end{thebibliography}

\bibliographystyle{elsarticle-num}
\end{document}